\pdfoutput=1
\documentclass[a4paper,numberwithinsect]{lipics-v2021}

\usepackage{amsfonts}
\usepackage{amsmath}
\usepackage{graphicx}
\usepackage{bussproofs}

\newcommand{\myComment}[1]{}

\newcommand{\defeq}{\stackrel{\mathsf{def}}{=}}
\newcommand{\defequiv}{\stackrel{\mathsf{def}}{\Longleftrightarrow}}
\newcommand{\set}[1]{\left\{ #1 \right\}}
\newcommand{\setdef}[2]{\left\{ #1 \>\middle|\> #2 \right\}}

\newcommand{\SOSTrans}[2]{#1 \Rightarrow #2}

\newcommand{\config}[2]{\left< #1, #2 \right>}

\newcommand{\SOSConfTrans}[4]{\SOSTrans{\config{#1}{#2}}{\config{#3}{#4}}}
\newcommand{\SOSConfFinTrans}[3]{\SOSTrans{\config{#1}{#2}}{#3}}

\newcommand{\denot}[1]{\left[\!\left[#1\right]\!\right]}
\newcommand{\fdenot}[1]{\left|\!\left|#1\right|\!\right|}
\newcommand{\arexdenot}[1]{\mathcal{A}\!\denot{#1}}
\newcommand{\bexdenot}[1]{\mathcal{B}\!\denot{#1}}

\newcommand{\trcdenot}[1]{\mathcal{S}_{\mathit{tr}}\!\denot{#1}}
\newcommand{\trcdenoth}[1]{\mathcal{S}^0_{\mathit{tr}}\!\denot{#1}}
\newcommand{\trcdenoti}[1]{\mathcal{S}_{\mathit{tr}}\!\denot{#1}_{\mathcal{I}}}
\newcommand{\sosdenot}[1]{\mathcal{S}_{\mathit{sos}}\!\denot{#1}}

\newcommand{\logdenot}[1]{\left|\!\left|#1\right|\!\right|_{\mathcal{V}}}
\newcommand{\logdenotV}[2]{\left|\!\left|#1\right|\!\right|_{#2}}
\newcommand{\logdenotsubst}[3]{\left|\!\left|#1\right|\!\right|_{\mathcal{V} [#2 \mapsto #3]}}

\newcommand{\restr}[2]{\!\!\mid_{#1}^{#2}}
\newcommand{\chop}[0]{\raisebox{1ex}{$\frown$}}

\newcommand{\wltf}[3]{\mathsf{wltf} (#1, #2, #3)}
\newcommand{\stf}[2]{\mathsf{stf} (#1, #2)}
\newcommand{\stfs}[1]{\mathsf{stf} (#1)}
\newcommand{\mes}[1]{\mathsf{mes} (#1)}
\newcommand{\stfsh}[1]{\mathsf{stf}' (#1)}

\newcommand{\stfh}[2]{\mathsf{stf}(#1, #2)}

\newcommand{\can}[1]{\mathsf{can} (#1)}

\newcommand{\WHILE}{\mbox{\bf While}}
\newcommand{\Rec}{\mbox{\bf Rec}}

\newcommand{\AExp}{\mbox{\bf AExp}}
\newcommand{\BExp}{\mbox{\bf BExp}}
\newcommand{\Stm}{\mbox{\bf Stm}}

\newcommand{\Skip}{\mbox{\bf skip}}

\newcommand{\If}{\mbox{\bf if}}
\newcommand{\Then}{\mbox{\bf then}}
\newcommand{\Else}{\mbox{\bf else}}
\newcommand{\While}{\mbox{\bf while}}
\newcommand{\Do}{\mbox{\bf do}}

\newcommand{\Any}{\mbox{\bf havoc}}
\newcommand{\Abort}{\mbox{\bf diverge}}

\newcommand{\ttt}{\mbox{\bf tt}}
\newcommand{\fff}{\mbox{\bf ff}}

\newcommand{\State}{\mbox{\bf State}}

\newcommand{\judge}[2]{#1 : #2}
\newcommand{\gjudge}[3]{#1 \rhd #2 : #3}

\newcommand{\sequent}[2]{#1 \>\vdash\> #2}
\newcommand{\semsequent}[2]{#1 \>\models\> #2}
\newcommand{\semsequenti}[2]{#1 \>\models_\mathcal{I}\> #2}

\newcommand{\Rule}[3]{
   \begin{tabular}{cc}
      {\sc {#1}}
        &
       \begin{tabular}{c}
          ${#2}$ \\ \hline
          ${#3}$ 
       \end{tabular}
    \end{tabular}}

\newcommand{\RuleTwo}[4]{
   \begin{tabular}{cc}
      {\sc {#1}}
        &
       \begin{tabular}{c}
         ${#2}$ \\
         ${#3}$ \\         \hline
          ${#4}$ 
       \end{tabular}
    \end{tabular}}

\newcommand{\dg}[1]{\marginpar{DG: #1}}

\hideLIPIcs
\nolinenumbers

\title{An Expressive Trace Logic for Recursive Programs}

\author{Dilian Gurov}{KTH Royal Institute of Technology, Stockholm, Sweden}{dilian@kth.se}{0000-0002-0074-8786}{}

\author{Reiner H\"{a}hnle}{Technical University of Darmstadt, Germany}{haehnle@cs.tu-darmstadt.de}{0000-0001-8000-7613}{}

\authorrunning{D. Gurov and R. H\"ahnle}

\Copyright{Dilian Gurov and Reiner H\"ahnle}



\begin{document}

\maketitle

\begin{abstract}
  We present an expressive logic over trace formulas, based on binary
  state predicates, chop, and least fixed-points, for precise specification 
  of programs with recursive procedures. Both, programs and trace
  formulas, are equipped with a direct-style, fully compositional,
  denotational semantics that on programs coincides with the standard
  SOS of recursive programs.
  We design a compositional proof calculus for proving finite-trace
  program properties, and prove soundness as well as (relative) completeness.
  We show that each program can be mapped to a semantics-preserving
  trace formula and, vice versa, each trace formula can be mapped to a
  canonical program over slightly extended programs, resulting in a
  Galois connection between programs and formulas.
  Our results shed light on the correspondence between programming
  constructs and logical connectives. 
\end{abstract}


\section{Introduction}
\label{sec:intro}

It is uncommon that specification languages used in program
verification are as expressive as the programs they are intended to
specify:
In model checking \cite{CGP99} one typically abstracts away from data
and unwinds unbounded control structures. Specification of unbounded
computations is achieved by recursively defined temporal operators.
In deductive verification \cite{HaehnleHuisman19}, first-order
Hoare-style contracts \cite{Hoare71} are the basis of widely used
specification languages \cite{acsl,JML-Ref-Manual}. The latter
specify computation by taking symbolic memory snapshots in terms of
first-order formulas, typically at the beginning and at the end of
unbounded control structures, such as procedure call and return, or
upon loop entry and exit. Contracts permit to approximate the effect
of unbounded computation as a finite number of first-order formulas
against which a program's behavior can be verified.

Imagine, in contrast, a logic for program specification that is
\emph{at least as expressive} as the programs it is supposed to
describe. Formulas $\phi$ of such a logic characterize a generally
infinite set of program computation traces $\sigma$. One can then form
\emph{judgments} of program statements $S$ of the form
$\judge{S}{\phi}$ with the natural semantics that any possible trace
$\sigma$ of $S$ is one of the traces described by $\phi$.

Two arguments against such a \emph{trace logic} are easily raised:
first, one expects a specification language to be capable of
\emph{abstraction} from implementation details; second, its
computational complexity. Regarding the first, we note that any
desired degree of abstraction can be achieved by \emph{definable}
constructs, i.e.~\emph{syntactic sugar}, in a sufficiently rich trace
logic. Regarding the second, it is far from obvious whether the
computational worst-case tends to manifest itself in \emph{practical}
verification \cite{HaehnleHuisman19}. First experiments seem to
indicate this is not the case.
%
On the other hand, a rich, trace-based specification logic bears many
advantages:
%
\begin{enumerate}
\item On the practical side, trace-based specification permits to
  express, for example, that an event or a call happened (exactly
  once, at least once) or that did not happen. Likewise, one can
  specify without reference to the target code (via assertions) that a
  condition holds at some point which is not necessarily an endpoint
  of a procedure call.  It is also possible to specify
  inter-procedural properties, such as call sequences, absence of
  callbacks, etc.  Such properties are important in security analysis
  and they are difficult or impossible to express with Hoare-style
  contracts or in temporal logic.
\item The semantics of trace formulas in terms of trace sets can be
  closely aligned with a suitable program semantics, which greatly
  simplifies soundness and completeness proofs.
\item An expressive trace logic makes for simple and natural
  completeness proofs for judgments $\judge{S}{\phi}$, because it is
  unnecessary to encode the program semantics in order to construct
  sufficiently strong first-order conditions.
\item In a calculus for judgments of the form $\judge{S}{\phi}$ one
  one can design \emph{inference} rules that directly decompose
  judgments into simpler ones, but also \emph{algebraic} rules that
  simplify and transform $\phi$, and one may mix both reasoning
  styles.
\item The semantics of programs, of trace formulas, and the rules of
  the calculus can be formulated in a \emph{fully compositional}
  manner: definitions and inference rules involve no intermediate
  state or context.
\item In consequence, we are able to establish a close
  correspondence between programs and formulas, which sheds light on
  the exact relation between each program construct and its
  corresponding logical connective. We formulate and prove a Galois
  connection that formalizes this fact.
\end{enumerate}

Our main contribution is 
a trace logic for imperative programs with recursive procedures where
we formalize and prove the claims above.
We restrict attention to the \emph{terminating executions} of
programs, i.e.~to their \emph{finite trace semantics}. This is not a
fundamental limitation, but the desire not to obscure the construction
with complications that may be added later.
Still, adaptating the theory to maximal traces (including infinite runs) 
is not trivial, and working out the details is left as future work
(see Section~\ref{sec:non-term-progr}).

Our paper is structured as follows: In Section~\ref{sec:rec} we define
the programming language \Rec\ used throughout this paper and we give
it a standard SOS semantics \cite{Plotkin04}. In
Section~\ref{sec:rec-traces} we define a denotational semantics for
\Rec\ in fully compositional style and prove it equivalent to the SOS
semantics. In Section~\ref{sec:logic-traces} we introduce our trace
logic and map programs to formulas by the concept of a \emph{strongest
  trace formula}, which is shown to fully preserve program
semantics. In Section \ref{sec:proof-calculus} we define a proof
calculus for judgments and show it to be sound and complete. In
Section \ref{sec:canonical-programs}, using the concept of a
\emph{canonical program}, we map trace formulas back to programs in
the slightly extended language $\Rec^\ast$ with non-deterministic
guards.  We prove that $\Rec^\ast$ and trace formulas form a Galois
connection via strongest trace formulas and canonical programs. In
Section~\ref{sec:related-work} we discuss related work, in
Section~\ref{sec:future-work} we sketch some extensions, including
options to render specifications more abstract and how to prove
consequence among trace formulas. In Section~\ref{sec:extro} we
conclude.


\section{The Programming Language \Rec}
\label{sec:rec}

We define a simple programming language \Rec\ with recursive
procedures and give it a standard SOS semantics.  We follow the
notation of~\cite{nie-nie-07-book}, adapted to \Rec\ syntax.
\begin{definition}[\Rec\ Syntax]
  \label{def:rec-syntax}
  A \Rec\ program is a pair $\langle S, T\rangle$, where $S$ is a
  statement with the \emph{abstract syntax} defined by the grammar:
  $$
  S \>::=\> \Skip \mid 
  x := a \mid 
  S_1 ; S_2 \mid 
  \If\ b\ \Then\ S_1\ \Else\ S_2 \mid 
  m()
  $$
  and where~$T$ is a \emph{procedure table} given by $T\,::=\,M^\ast$ with
  declarations of parameter-less procedures $M\,::=\,m\,\{S\}$.

  In the grammar, $a$ ranges over arithmetic expressions $\AExp$, and
  $b$ over Boolean expressions $\BExp$. Both are assumed to be
  side-effect free and are not specified further. All variables are
  global and range over the set of integers~$\mathbb{Z}$. We assume
  programs are \emph{well-formed}, i.e., only declared procedures are
  called, and procedure names are unique.
\end{definition}

\begin{example}
\label{ex:rec-program}
  Consider the \Rec\ program $p_0\defeq\langle S, T\rangle$ with statement
  $S\defeq x := 3; \mathit{even} ()$ and procedure table:
  $$ T\quad\defeq\quad
  \begin{array}{r@{\;}l}
    \mathit{even} & \set{\If\ x = 0\ \Then\ y := 1\ \Else\ x := x-1; \mathit{odd} ()} \\
    \mathit{odd} & \set{\If\ x = 0\ \Then\ y := 0\ \Else\ x := x-1; \mathit{even} ()}
  \end{array}
  $$
  The intended semantics is that $\mathit{even}$ terminates in a state
  where $y=1$ if and only if it is started in a state where $x$ is
  even and non-negative, and it terminates in a state where $y=0$ if
  and only if it is started in a state where $x$ is odd and
  non-negative.
\end{example}

\begin{example}
\label{ex:rec-program-simple}
  Consider the \Rec\ program $p_1\defeq\langle S, T\rangle$ with
  $S\defeq\mathit{down}()$ and procedure table:
  $$ T\quad\defeq\quad
  \mathit{down} 
  \set{\If\ x > 0\ \Then\ x := x-2; \mathit{down}()\ \Else\ \Skip}
  $$
  The intended semantics is that $p_1$ terminates in a state where
  $x=0$ if and only if it is started in state where $x$ is even and
  non-negative.
\end{example}

\begin{remark}\label{rem:while-syn}
  Loops can be defined with the help of (tail-)recursive programs.
  For example, a loop of the form ``$\While\ b\ \Do\ S$'' can be
  simulated with a procedure declared in $T$ as:
  $$
  m\,\{\If\ b\ \Then\ S;m()\ \Else\ \Skip\}
  $$
  using a unique name $m$ and replacing the occurrence of the loop
  with a call to $m()$.
\end{remark}

\begin{figure}
\centering
\Rule{Skip}
  {-}
  {\SOSConfFinTrans{\Skip}{s}{s}} \quad
\Rule{Assign}
  {-}
  {\SOSConfFinTrans{x := a}{s}{s [x \mapsto \arexdenot{a} (s)]}} \\~\\

\Rule{Seq-1}
  {\SOSConfFinTrans{S_1}{s}{s'}}
  {\SOSConfTrans{S_1 ; S_2}{s}{S_2}{s'}} \quad
\Rule{Seq-2}
  {\SOSConfTrans{S_1}{s}{S_1'}{s'}}
  {\SOSConfTrans{S_1 ; S_2}{s}{S_1' ; S_2}{s'}} \\~\\

\Rule{If-1}
  {-}
  {\SOSConfTrans{\If\ b\ \Then\ S_1\ \Else\ S_2}{s}{S_1}{s}}~
  \mbox{if $\bexdenot{b} (s) = \ttt$} \\~\\

\Rule{If-2}
  {-}
  {\SOSConfTrans{\If\ b\ \Then\ S_1\ \Else\ S_2}{s}{S_2}{s}}~
  \mbox{if $\bexdenot{b} (s) = \fff$} \\~\\

\Rule{Call}
  {-}
  {\SOSConfTrans{m()}{s}{S}{s}}~
  \mbox{if $m$ is declared as $m\,\{S\}$ in $T$}
\caption{\label{fig:rec-sos}SOS rules for \Rec.}
\end{figure}

A standard, \emph{structural operational semantics} (SOS) for \Rec\ is
defined in Figure~\ref{fig:rec-sos} (sometimes referred to as
\emph{small-step semantics}).
We use it as a base line when defining the denotational finite-trace
semantics in Section~\ref{sec:rec-traces}.

Let $\mathit{Var}$ be the set of program variables, and \State\ the
set of program states $s : \mathit{Var} \rightarrow \mathbb{Z}$.
Let $\arexdenot{a} (s) \in \mathbb{Z}$ denote the integer value of the 
arithmetic expression~$a$ when evaluated in state~$s$, and
$\bexdenot{b} (s) \in \mathbb{T}$ denote the truth value of the Boolean
expression~$b$ when evaluated in state~$s$, both defined as expected.

A \emph{configuration} is either a pair $\config{S}{s}$ consisting of
a statement and a state, designating an initial or intermediate
configuration; or a state~$s$, designating a final configuration.  To
simplify notation we assume that $S$ is evaluated relative to a \Rec\
program with a procedure table $T$ which is not explicitly specified.

The transitions of the SOS either relate two intermediate configurations,
or an intermediate with a final one, and thus have the shape
$\SOSConfTrans{S}{s}{S'}{s'}$ or $\SOSConfFinTrans{S}{s}{s'}$, 
respectively.

\begin{definition}[\Rec\ SOS]
  \label{def:rec-sos}
  The \emph{structural operational semantics} (SOS) of \Rec\ is
  defined by the rules given in Figure~\ref{fig:rec-sos}.
\end{definition}


The structural operational semantics of \Rec\ \emph{induces} a
finite-trace semantics in terms of the sequences of states that are
traversed from an initial to a final configuration when executing a
given statement in the SOS.
Let $\State^+$ denote the set of all non-empty, finite sequences of
states. Formally, we define a function
$\sosdenot{S} : \Stm \rightarrow 2^{\State^+}$, i.e., a function
such that $\sosdenot{S} \subseteq \State^+$ for any statement~$S$.

\begin{definition}[Induced Finite-Trace Semantics]
  \label{def:induced-finite-trace-semantics}
  Let $S$ be a statement. Then, $\sosdenot{S}$ is defined as the set
  of (finite) sequences $s_0 \cdot s_1 \cdot \ldots \cdot s_n$ of
  states for which there are statements $S_0, S_1, \ldots, S_{n-1}$
  such that $S_0 = S$, $\SOSConfTrans{S_i}{s_i}{S_{i+1}}{s_{i+1}}$ for
  all $0 \leq i \leq n-2$, and
  $\SOSConfFinTrans{S_{n-1}}{s_{n-1}}{s_n}$.
\end{definition}
Observe that in Definition~\ref{def:induced-finite-trace-semantics}
\emph{any} state $s_0 \in \State$ can serve as the initial state of a
finite trace.
Next we design a ``direct-style'', denotational finite-trace semantics
that \emph{conforms} with the SOS, in the sense that it is equal to the
finite-trace semantics induced by the SOS.


\section{A Denotational Finite-Trace Semantics for \Rec}
\label{sec:rec-traces}

We define the semantic function
$\mathcal{S}_{\mathit{tr}} : \Stm \rightarrow 2^{\State^+}$ with the
intention that $ \trcdenot{S} \>=\> \sosdenot{S} $.
%
Unlike $\mathcal{S}_{sos}$, however, $\mathcal{S}_{\mathit{tr}}$ is
defined directly, without referring to other semantic rules as SOS
does (hence the term ``direct-style'').
%
%
We  define $\trcdenot{S}$ in the style of \emph{denotational 
semantics}, compositionally, by induction on the structure of~$S$, 
and through defining equations.



Let us define a unary \emph{restriction} operator on trace sets, 
for any trace set~$A$ and Boolean expression~$b \in \BExp$, as follows:
$
\begin{array}{@{}l@{\,}c@{\,}l@{}}
  A \restr{b}{} & \>\defeq\> & \setdef{s \cdot \sigma \in A}{\bexdenot{b} (s) = \ttt}
\end{array}
$.
It filters out all traces in~$A$ whose first state does not
satisfy~$b$.
Another unary operator on trace sets is defined as
$
\sharp A \>\defeq\> 
\setdef{s \cdot s \cdot \sigma}{s \cdot \sigma \in A}
$
which duplicates the first state in each trace in~$A$.
Finally, let us define the binary operator on trace sets:
$ A \chop B \>\defeq\> \setdef{\sigma_A \cdot s \cdot \sigma_B}
{\sigma_A \cdot s \in A \>\wedge\> s \cdot \sigma_B \in B} $
which concatenates traces from~$A$ with traces from~$B$ that agree on
the last and first state, respectively, but without duplicating that
state, see also \cite{HalpernShoham91}.




\begin{figure}
  \centering
$ \begin{array}{@{}r@{\;}c@{\;}l@{\qquad}r@{\;}c@{\;}l@{}}
   \trcdenoth{\Skip}_\rho  & \defeq & 
      \setdef{s \cdot s}{s \in \State} &
   \trcdenoth{x := a}_\rho & \defeq & 
      \setdef{s \cdot s[x \mapsto \arexdenot{a} (s)]}{s \in \State} \\
   \trcdenoth{S_1 ; S_2}_\rho  & \defeq & 
      \trcdenoth{S_1}_\rho \chop \trcdenoth{S_2}_\rho & 
   \trcdenoth{m_i ()}_\rho & \defeq & \rho(m_i) \\ 
    \multicolumn{6}{c}{%
    \trcdenoth{\If\ b\ \Then\ S_1\ \Else\ S_2}_\rho   \,\defeq\,  
      (\sharp \trcdenoth{S_1}_\rho) \!\!\mid_b \cup\, 
      (\sharp \trcdenoth{S_2}_\rho) \!\!\mid_{\neg b}}      \\
   \end{array} $
\caption{Finite-trace semantic equations for \Rec.}
\label{fig:while-tr}
\end{figure}

\begin{definition}[Denotational Finite-Trace Semantics of \Rec]
\label{def:while-tr}
Let~$M = \set{m_1,\ldots,m_n}$ be the set of procedure names in \Rec\
program~$\langle S,T\rangle$, where every~$m_i$ is declared as
$m_i\,\{S_i\}$ in~$T$.
We define a helper function $\trcdenoth{S}_\rho$ that is relativized
on an interpretation $\rho : M \rightarrow 2^{\State^+}$ of the
procedure names, inductively, by the equations given in
Figure~\ref{fig:while-tr}.
The duplication of the initial states in the equation for the $\If$
statement is needed to remain faithful to the SOS of $\Rec$, which
allocates a small step for evaluation of the guard~$b$.
We then introduce a semantic function
$ H : \bigl(2^{\State^+}\bigr)^n \rightarrow \bigl(2^{\State^+}\bigr)^n $ 
defined as:
$$
H (\rho) \,\defeq\, (\sharp \trcdenoth{S_1}_\rho, \ldots, \sharp \trcdenoth{S_n}_\rho)\enspace .
$$

Function $H$ is monotonic and continuous in the CPO with bottom
$\bigl(\bigl(2^{\State^+}\bigr)^n, \sqsubseteq, \varnothing^n\bigr)$,
where $\sqsubseteq$~denotes point-wise set inclusion.  Hence, by the
Knaster-Tarski Theorem, it has a least fixed-point.
Let~$\rho_0$ denote this least fixed-point.
The \emph{denotational finite-trace semantics} of statements~$S$
of \Rec\ is defined relative to this inter\-pretation as:
$ \trcdenot{S} \,\defeq\, \trcdenoth{S}_{\rho_0} $.
\end{definition}



The finite-trace semantics agrees with the SOS of $\Rec$, in the sense
that $\trcdenot{S}$ coincides with the finite-trace semantics
$\sosdenot{S}$ induced by the SOS, as defined in
Definition~\ref{def:induced-finite-trace-semantics}.

\begin{theorem}[Correctness of Trace Semantics]
\label{thm:semantics-correctness-rec}
For all statements~$S$ of \Rec, we have:
$ \trcdenot{S} \>=\> \sosdenot{S} $
\end{theorem}


\section{A Logic over Finite Traces}
\label{sec:logic-traces}

Our trace logic can be seen as an Interval Temporal Logic
\cite{HalpernShoham91} with $\mu$-recursion \cite{Park76}, or
alternatively, as a temporal $\mu$-calculus with a binary temporal
operator corresponding to the chop operation over sets of traces (see,
e.g., \cite{sti-93-handbook-chapter} for a general introduction to
$\mu$-calculus).


\subsection{Syntax and Semantics of the Logic}

The philosophy behind our logic is to have logical counterparts to 
the statements of the programming language in terms of their 
finite-trace semantics. 
For instance, we use binary relation symbols that correspond to the
atomic statements, and a chop operator corresponding to sequential 
composition.
This design choice helps to simplify proofs of the properties of the
logic and the calculus.

\begin{definition}[Logic Syntax]
\label{def:logic-syntax}
The \emph{syntax} of the logic of \emph{trace formulas} is defined by
the following grammar:
%
$$ \phi \>::=\> p \mid
                R \mid 
                X \mid 
                \phi_1 \wedge \phi_2 \mid 
                \phi_1 \vee \phi_2 \mid 
                \phi_1 \chop \phi_2 \mid 
                \mu X. \phi
$$
where $p$~ranges over state formulas not further specified here, but
assumed to contain at least the Boolean expressions~$\BExp$,
$R$~ranges over binary relation symbols over states, and $X$ over a
set $\mathsf{RVar}$ of \emph{recursion variables}.
\end{definition}

\begin{figure}
  \centering
  $ \begin{array}{rcl@{\qquad}rcl}
   \logdenot{p}  & \defeq & 
      \setdef{s \cdot \sigma}{s \models p} &
   \logdenot{R}  & \defeq & 
      \setdef{s \cdot s'}{R (s, s')} \\
   \logdenot{X}  & \defeq & 
      \mathcal{V} (X) &
   \logdenot{\phi_1 \wedge \phi_2}  & \defeq & 
      \logdenot{\phi_1} \cap \logdenot{\phi_2} \\
   \logdenot{\phi_1 \vee \phi_2}  & \defeq & 
      \logdenot{\phi_1} \cup \logdenot{\phi_2} &
   \logdenot{\phi_1 \chop \phi_2}  & \defeq & 
      \logdenot{\phi_1} \chop \logdenot{\phi_2} \\
   \multicolumn{6}{c}{\logdenot{\mu X. \phi}  \ \defeq \ 
      {\displaystyle\bigcap} \setdef{\gamma \subseteq \State^+}
                     {\logdenotsubst{\phi}{X}{\gamma} \subseteq \gamma}}
   \end{array} $
\caption{Finite-trace semantic equations for formulas.}
\label{fig:logic-semantics}
\end{figure}

\begin{definition}[Logic Semantics]
  \label{def:logic-semantics}
  The \emph{finite-trace semantics} of a formula~$\phi$ is defined as
  its denotation $\logdenot{\phi} \subseteq \State^+$, relativized on
  a \emph{valuation}
  $\mathcal{V} : \mathsf{RVar} \rightarrow 2^{\State^+}$ of the
  recursion variables, inductively by the equations given in
  Figure~\ref{fig:logic-semantics}, where in the last clause
  $\mathcal{V} [X \mapsto \gamma]$ denotes the updated valuation.
\end{definition}

One can show that the transformers
$\lambda \gamma.\logdenotsubst{\phi}{X}{\gamma}$ are monotonic
functions in the complete lattice $(2^{\State^+}, \subseteq)$ and
hence, by Tarski's fixed-point theorem for complete
lattices~\cite{tar-53}, they have least and greatest fixed-points.
In particular, the least fixed point is simultaneously also the least
pre-fixed point, hence the defining equation for $\mu X. \phi$.
And because it is a fixed point, we have the following result for
unfolding fixed-point formulas.

\begin{proposition}[Fixed-Point Unfolding]
  \label{prop:fixed-point-unfolding}
  Let $\mu X. \phi$ be a formula and $\mathcal{V}$ a valuation.  Then:
  $
  \logdenot{\mu X. \phi} = \logdenot{\phi [\mu X. \phi / X]}
  $.
\end{proposition}

Our calculus is based on closed formulas of the logic.  Observe that
fixed-point unfolding preserves closedness.  For closed formulas the
valuation~$\mathcal{V}$ is immaterial to the semantics
$\logdenot{\phi}$. In this case, we often omit the subscript and
simply write~$\fdenot{\phi}$.


\subsection{Binary Relations}
\label{sec:binary-relations}

We instantiate the set~$\mathsf{Rel}$ of binary relation symbols with
two specific relations:
$$
\begin{array}{rcl}
  \mathit{Id} (s, s')  &  
                         \>\defequiv\>  &  
                                          s' = s \\
  \mathit{Sb}^{a}_{x} (s, s')  &  
                                 \>\defequiv\>  &  
                                                  s' = s[x \mapsto \arexdenot{a} (s)]
\end{array}
$$

These two relations are used to model $\Skip$ and assignment
statements, respectively.
%
The \emph{transitive closure}~$R^+$ of a binary relation~$R$ over
states is easily defined as a recursive formula in our logic as
$ R^+ \>\defequiv\> \mu X.\ (R \vee R \chop X) $.


\begin{example}
  \label{ex:form-1}
  For any arithmetic expression $a$ let $\mathit{Dec}_a$ be a binary
  relation symbol interpreted as follows:
  $
  \mathit{Dec}_a (s, s')  \>\defequiv\> 
  \arexdenot{a} (s') \leq \arexdenot{a} (s)
  $.
  That is, the value of~$a$ does not increase between two consecutive
  states.
  With this symbol, the formula $\mathit{Dec}_a^+$ expresses the
  property that the value of~$a$ does not increase throughout the
  whole execution of a program.
\end{example}

\begin{remark}
  It is conceivable to define arbitrary binary relations $(p,q)$ over
  state formulas $p$, $q$, i.e., a pair of pre- and post-condition.
  Alternatively, one can view a binary relation as a TLA
  \emph{action}~\cite{lam-94-tla}, and use primed versions of the
  variables to refer to their values in the second state of a pair.
  We do not require this generality in our examples.
\end{remark}


\begin{figure}
  \centering
    $\stfh{\overline{X}}{\Skip}  \defeq \mathit{Id} \qquad
    \stfh{\overline{X}}{x := a}  \defeq \mathit{Sb}^{a}_x \qquad
    \stfh{\overline{X}}{S_1;\,S_2}  \defeq \stfh{\overline{X}}{S_1} \chop \stfh{\overline{X}}{S_2}
    $
    \begin{align*}
    \stfh{\overline{X}}{\If\ b\ \Then\ S_1\ \Else\ S_2} & \defeq 
       (b \land \mathit{Id} \chop \stfh{\overline{X}}{S_1})
       \lor
       (\lnot b \land \mathit{Id} \chop \stfh{\overline{X}}{S_2}) \\
  \stfh{\overline{X}}{m()} & \defeq 
  \left\{
    \begin{array}{l@{\quad}l}
    \mathit{Id} \chop \mu X_m.\, \stfh{\overline{X}\cup\{m\}}{S_m} & 
    ~~m\not\in\overline{X},\,m\,\{{S_m}\}\in T \\
    \mathit{Id} \chop X_m & ~~\text{otherwise}
    \end{array}
          \right.
  \end{align*}
\caption{\label{fig:stf}Definition of strongest trace formula.}
\end{figure}


\subsection{Strongest Trace Formulas} 
\label{sec:stf}

Since our program logic is able to characterize program traces, and not
merely pre- and postconditions or intermediate assertions, it is
possible to establish a close correspondence between programs and
trace formulas.
This correspondence is captured by the
following---constructive---definition of the \emph{strongest trace
formula} $\stfs{S}$ for a given program~$S$, which characterizes all
terminating traces of~$S$.

For each procedure declaration $m\,\{S_m\}$ in $T$, we create a
fixed-point formula, whenever $m$ is called the first time. Subsequent
calls to $m$ result in a recursion variable. To achieve this, we
parameterize the strongest trace formula function with the already
created recursion variables $\overline{X}$. This parameter is
initialized to~$\varnothing$ and is ignored by all case definitions
except the one for a recursive call.

\begin{definition}[Strongest Trace Formula]
\label{def:stf}
  Let $\langle S,T\rangle$ be a \Rec\ program.
  The \emph{strongest trace formula} for~$S$, denoted
  $\stfs{S}$, is defined as $\stfs{S}\defeq\stfh{\varnothing}{S}$, where
  $\stfh{X}{S}$ is defined inductively in Figure~\ref{fig:stf}.
\end{definition}

\begin{example}\label{ex:stf-rec}
  For the program $\mathit{even}()$ with the procedure table of
  Example~\ref{ex:rec-program}, the strongest trace formula is:

  \noindent
  $
  \begin{array}{@{}l@{}l}
    & \mathit{Id} \chop \mu X_\mathit{even}.\Bigl( 
    \bigl(x = 0 \,\land\,\mathit{Id} \chop \mathit{Sb}^{1}_{y}\bigl)\ \lor\\ 
   & \;\bigl(x \neq 0 \land
     \mathit{Id} \chop \mathit{Sb}^{x-1}_{x} \chop \mathit{Id} \chop \mu X_{\mathit{odd}}.((x = 0 \land \mathit{Id} \chop \mathit{Sb}^{0}_{y}) \lor
  (x \neq 0 \land 
     \mathit{Id} \chop \mathit{Sb}^{x-1}_{x} \chop\mathit{Id}\chop X_{\mathit{even}}))\bigr)\Bigr)
\end{array}
$

\noindent The binder for $X_{odd}$ can be removed without changing the
semantics.
\end{example}

\begin{example}\label{ex:stf-rec-simple}
  For the program in Example~\ref{ex:rec-program-simple}, the
  strongest trace formula is:\smallskip
  
  \noindent$\stfs{S} = 
      \mathit{Id} \chop 
      \mu X_\mathit{down}.\, 
          \bigl((x > 0 \land \mathit{Id} \chop 
                 \mathit{Sb}^{x-2}_x 
                 \chop \mathit{Id} \chop X_\mathit{down})
                \lor
                 (x \leq 0 \land \mathit{Id}\chop\mathit{Id})
                  \bigr)$ 
\end{example}

\begin{theorem}[Characterisation of Strongest Trace Formula]
\label{thm:stf-characterisation-rec}
Let $\langle S,T\rangle$ be a program of \Rec. Then the following holds:
$ \fdenot{\stfh{\varnothing}{S}} \>=\> \trcdenot{S} $
\end{theorem}

\section{A Proof Calculus}
\label{sec:proof-calculus}

We present a proof calculus for our logic in the form of a
Gentzen-style deductive proof system, which is \emph{compositional}
both in the statement and the formula.




\subsection{Definition of the Calculus}
\label{sec:proof-calculus-rec}

To obtain a compositional proof rule for procedure calls, its shape
will essentially embody the principle of fixed-point induction
explained in the Appendix. 
For this we need to represent recursion variables in the $\Rec$
language, whose syntax is extended with a set $\mathsf{SVar}$ of
\emph{statement variables}, ranged over by~$Y$. We add these as a new
category of atomic statements to $\Rec$.


To define the semantics of programs in the presence of statement
variables, we relativize the finite-trace semantics $\trcdenoti{S}$ of
statements~$S$ on \emph{interpretations}
$\mathcal{I} : \mathsf{SVar} \rightarrow 2^{\State^+}$ of the
statement variables, lifted from $\trcdenot{S}$ in the canonical
manner.


\begin{definition}[Calculus Syntax]
  \label{def:calculus-syntax}
  \emph{Judgments} are of the form $\judge{S}{\phi}$, where $S$ is a
  $\Rec$ statement, possibly containing statement variables, and
  $\phi$ is a closed trace formula.
  The \emph{sequents} of the calculus are of the shape
  $\sequent{\Gamma}{\judge{S}{\phi}}$, where $\Gamma$ is a possibly
  empty set of judgments.
\end{definition}


\begin{figure}
\centering
\Rule{Skip}
{-}
{\sequent{\Gamma}{\judge{\Skip}{\mathit{Id}}}}~~~~~
\quad
\Rule{Assign}
{-}
{\sequent{\Gamma}{\judge{x := a}{\mathit{Sb}^{a}_{x}}}} 
\\~\\

\Rule{Seq}
{\sequent{\Gamma}{\judge{S_1}{\phi_1}} \quad
  \sequent{\Gamma}{\judge{S_2}{\phi_2}}}
{\sequent{\Gamma}{\judge{S_1; S_2}{\phi_1 \chop \phi_2}\rule{0ex}{2.5ex}}}
\quad
\Rule{If}
{\sequent{\Gamma}{\judge{\Skip; S_1}{\lnot b \lor \phi}} \quad
  \sequent{\Gamma}{\judge{\Skip; S_2}{b \lor \phi}}}
{\sequent{\Gamma}{\judge{\If\ b\ \Then\ S_1\ \Else\ S_2}{\phi}}}
\\~\\

\Rule{Unfold}
{\sequent{\Gamma}{\judge{S}{\phi [\mu X. \phi / X]}}}
{\sequent{\Gamma}{\judge{S}{\mu X. \phi}}} ~~~~~
\qquad
\Rule{Cons}
{\sequent{\Gamma}{\judge{S}{\phi'}}}
{\sequent{\Gamma}{\judge{S}{\phi}}}~~
$\phi' \models \phi$  
\\~\\

\RuleTwo{Call}
  {\judge{Y_m}{\phi_m}\not\in\Gamma\qquad m\,\{S_m\}\in T}
  {\sequent{\Gamma,\judge{Y_m}{\phi_m}}{\judge{S_m[\Skip;Y_m/m(),\Skip;Y_{m_1}/m_1(),\ldots,\Skip;Y_{m_n}/m_n()]}{\phi_m}}}
  {\sequent{\Gamma}{\judge{m()}{\mathit{Id}\chop\phi_m}}}
%
\caption{The rules of the proof calculus.}
\label{fig:calculus}
\end{figure}

\subsubsection{Rules}

The calculus has exactly one rule for each kind of \Rec\ statement,
except for statement variables which do not occur in initial judgments
to be proven, but are only created intermittently in proofs by the
(\textsc{Call}) rule. All statement rules are \emph{compositional} in
the sense that only the statement $S$ in focus without any context
appears in the conclusion.

The statement rules and two selected logical rules of the calculus are
shown in Figure~\ref{fig:calculus}. The remaining logical rules, in
particular the axioms, are the standard Gentzen-style ones and are
omitted.

To prove the judgment $\judge{S}{\phi}$ for a program
$\langle S,T\rangle$, we prove the sequent
$\sequent{}{\judge{S}{\phi}}$. All rules, except the (\textsc{Call})
rule, leave the antecedent $\Gamma$ invariant.

We first explain the two logical rules.
The (\textsc{Unfold}) rule is based on
Proposition~\ref{prop:fixed-point-unfolding} and is used to unfold
fixed-point formulas.
%
The consequence rule (\textsc{Cons}) permits to strengthen the trace
formula~$\phi$ in the succedent, i.e., the specification of the
program under verification. This is typically required to achieve a
suitable syntactic form of~$\phi$, or to strengthen an inductive
claim.  The rule assumes the existence of an \emph{oracle} for proving
the logical entailment between trace formulas.

The (\textsc{Skip}) and (\textsc{Assign}) rules handle the atomic
statements, using the two binary relation symbols defined in
Section~\ref{sec:binary-relations}.
The (\textsc{Seq}) rule is a rule for sequential composition. Observe
that it is compositional in the sense that \emph{no intermediate
  state} between $S_1$ and $S_2$ needs to be considered.

The (\textsc{If}) rule is a compositional rule for conditional
statements. The trace formulas $\lnot b\lor\phi$ in the left premise
(and $b\lor\phi$ in the right one) might at first appear
counter-intuitive. Formula $\lnot b\lor\phi$ is read as follows: We
need not consider program $S_1$ for any trace, where $\lnot b$ holds
in the beginning, because these traces relate to $S_2$; otherwise,
$\phi$ must hold. A more intuitive notation would be
$b\rightarrow\phi$, but we refrain from introducing implication in our
logic.

Unsurprisingly, the (\textsc{Call}) rule is the most complex.  We
associate with each declaration of a method~$m$ in~$T$ a unique
statement variable $Y_m$. The antecedent $\Gamma$ contains judgments
of the form $\judge{Y_m}{\phi_m}$.  One can think of the~$Y_m$ as a
generic continuation of any recursive call to~$m$ of which we know
that it must conform to its contract~$\phi_m$. Once this conformance
has been established, the fact is memorized in the antecedent
$\Gamma$. Therefore, the (\textsc{Call}) rule is triggered only when a
call to procedure $m()$ is encountered the first time.  This is
ensured by the condition $\judge{Y_m}{\phi_m}\not\in\Gamma$. To avoid
having to apply the call rule again to recursive calls of $m()$, all
such calls in the body $S_m$ are replaced with $\Skip;Y_m$, where the
$\Skip$ models unfolding and $Y_m$ is justified by the assumption in
$\Gamma$. Likewise, any other $\judge{Y_{m_i}}{\phi_{m_i}}\in\Gamma$
triggers an analogous substitution. Now the procedure body
$S_m[\ldots]$ in the premise contains at most procedure calls to $m'$
that do not occur in $\Gamma$.



If a different judgment than $\mathit{Id}\chop\phi_m$ is to be proven,
then rule (\textsc{Cons}) must be applied before (\textsc{Call}) to
achieve the required shape.

\def\defaultHypSeparation{\hskip 2em}

\begin{figure*}
    \begin{prooftree}
      \AxiomC{$\sequent{\cdots}{\judge{\Skip}{\mathit{Id}}$}}
      \AxiomC{$\sequent{\judge{Y_{\mathit{even}}}{\phi_{\mathit{even}}},\judge{Y_{\mathit{odd}}}{\phi_{\mathit{odd}}}}{\judge{Y_{\mathit{even}}}{\phi_{\mathit{even}}}}$}
      \LeftLabel{\textsc{Cons}}
      \UnaryInfC{$\sequent{\judge{Y_{\mathit{even}}}{\phi_{\mathit{even}}},\judge{Y_{\mathit{odd}}}{\phi_{\mathit{odd}}}}{\judge{Y_{\mathit{even}}}{\mu X_{\mathit{even}}.(\cdots\chop\phi_{\mathit{odd}})}}$}
      \LeftLabel{\textsc{Seq}}
      \BinaryInfC{$\sequent{\judge{Y_{\mathit{even}}}{\phi_{\mathit{even}}},\judge{Y_{\mathit{odd}}}{\phi_{\mathit{odd}}}}{\Skip;\judge{Y_{\mathit{even}}}{\mathit{Id}\chop\mu X_{\mathit{even}}.(\cdots\chop\phi_{\mathit{odd}})}}$}
      \UnaryInfC{$\vdots$}
      \UnaryInfC{$\sequent{\judge{Y_{\mathit{even}}}{\phi_{\mathit{even}}},\judge{Y_{\mathit{odd}}}{\phi_{\mathit{odd}}}}{\judge{S_{\mathit{odd}}[\Skip;Y_{\mathit{even}}/\mathit{even}()]}{\phi_{\mathit{odd}}}}$}
      \LeftLabel{\textsc{Call}}
      \UnaryInfC{$\sequent{\judge{Y_{\mathit{even}}}{\phi_{\mathit{even}}}}{\judge{\mathit{odd}()}{\stfs{\mathit{odd}()}}}$}
      \LeftLabel{\textsc{Cons}}
      \UnaryInfC{$\sequent{\judge{Y_{\mathit{even}}}{\phi_{\mathit{even}}}}{\judge{\mathit{odd}()}{\mathit{Id}\chop\mu X_{\mathit{odd}}.((x = 0 \land \mathit{Id} \chop \mathit{Sb}^{0}_{y}) \lor(x \neq 0 \land 
     \mathit{Id} \chop \mathit{Sb}^{x-1}_{x}\chop\mathit{Id}\chop\phi_{\mathit{even}})}}$}
      \UnaryInfC{$\vdots$}
      \UnaryInfC{$\sequent{\judge{Y_{\mathit{even}}}{\phi_{\mathit{even}}}}{\judge{S_{\mathit{even}}}{\phi'_{\mathit{even}}[\phi_{\mathit{even}}/X_{\mathit{even}}]}}$}
      \LeftLabel{\textsc{Unfold}}
      \UnaryInfC{$\sequent{\judge{Y_{\mathit{even}}}{\phi_{\mathit{even}}}}{\judge{S_{\mathit{even}}}{\phi_{\mathit{even}}}}$}
      \LeftLabel{\textsc{Call}}
      \UnaryInfC{$\sequent{}{\judge{\mathit{even}()}{\stfs{\mathit{even}()}}}$}
    \end{prooftree}
  \caption{\label{fig:proof-rec}Proof of
    $\judge{\mathit{even}()}{\stfs{\mathit{even}()}}$.}
\end{figure*}

\begin{example}\label{ex:proof-rec}
  We prove the judgment
  $\judge{\mathit{even}()}{\stfs{\mathit{even}()}}$ for the program
  from Example~\ref{ex:rec-program} in Figure~\ref{fig:proof-rec}. We
  abbreviate the fixed-point formula in $\stfs{\mathit{even}()}$ from
  Example~\ref{ex:stf-rec} with $\phi_{\mathit{even}}$, so that
  $\stfs{\mathit{even}()}=\mathit{Id}\chop\phi_{\mathit{even}}$, the
  body of $\mathit{even}()$ with $S_{\mathit{even}}$, and similarly
  for $\mathit{odd}()$. Moreover, we abbreviate
  \begin{multline*}
    \phi'_{\mathit{even}}=\bigl(x = 0\land\mathit{Id}\chop \mathit{Sb}^{1}_{y}\bigr)\lor
    \bigl(x \neq 0\,\land\\
    \qquad\mathit{Id} \chop \mathit{Sb}^{x-1}_{x} \chop \mathit{Id}\chop\mu X_{\mathit{odd}}.((x = 0 \land \mathit{Id} \chop \mathit{Sb}^{0}_{y}) \lor(x \neq 0 \land 
    \mathit{Id} \chop \mathit{Sb}^{x-1}_{x}\chop \mathit{Id} \chop X_{\mathit{even}})\bigr)\bigr)
  \end{multline*}
  
  The proof starts with the call rule, followed by an unfold of the
  fixed point formula $\phi_{\mathit{even}}$. We now proceed with the
  other statement rules simultaneously on $S_{\mathit{even}}$ and the
  formula on the right until we encounter the call of $\mathit{odd}()$
  in $S_{\mathit{even}}$. Here we would like to apply the call rule to
  $\mathit{odd}()$, but we do not have $\phi_{\mathit{odd}}$ on the
  right, because the unfolding of $\phi_{\mathit{even}}$ went ``too
  deep''. To avoid a lengthy derivation at this point, we use the fact
  that trace formula on the right is \emph{equivalent} to
  $\phi_{\mathit{odd}}$, and use the consequence rule to obtain it.

  Now we descend into $\mathit{odd}()$, similarly as before; however,
  because the judgment
  $\judge{Y_{\mathit{even}}}{\phi_{\mathit{even}}}$ is present on the
  left, the call rule replaces $\mathit{even}()$ in
  $S_{\mathit{even}}$ with
  $\Skip;\judge{Y_{\mathit{even}}}{\phi_{\mathit{even}}}$. Finally, we
  encounter the statement variable $Y_{\mathit{even}}$, but again the
  fixed point formula on the right is ``too deep''. After a second
  transformation we close the proof with an axiom.
\end{example}

\begin{remark}\label{rem:while-rule}
  It is easy to derive a rule for loops from the (\textsc{Call}) rule
  using the encoding $m\,\{\If\ b\ \Then\ S;m()\ \Else\ \Skip\}$ given
  in Remark~\ref{rem:while-syn}. In a pure loop program no unprocessed
  recursive call except $m()$ ever occurs in the body, so the
  \textsc{Call}) rule is applicable with $\Gamma=\varnothing$.  Rule
  (\textsc{Call}) instantiated to~$m$ and a suitable~$\phi_m$ gives
  $ S_m[\Skip;Y_m/m()]=\If\ b\ \Then\ S;\Skip;Y_m\ \Else\ \Skip $ in
  the premise on the right, so its single premise becomes:
  $ \sequent{\judge{Y_m}{\phi_m}}{\judge{\If\ b\ \Then\ S;\Skip;Y_m\
      \Else\ \Skip}}{\phi_m} $.  Subsequent application of rule
  (\textsc{If}) yields the two premises
  $ \sequent{\judge{Y_m}{\phi_m}}{\judge{\Skip;S;\Skip;Y_m}}{\lnot
    b\lor\phi_m} \text{ and }
  \sequent{\judge{Y_m}{\phi_m}}{\judge{\Skip;\Skip}}{b\lor\phi_m}.  $.
  In each premise is a spurious $\Skip$ resulting from evaluating the
  method call which is only due to the encoding. In addition, the
  antecedent is not needed to prove the second premise and can be
  removed. After reordering and simplification, rule (\textsc{While})
  is obtained as:
  \[
    \Rule{While}
    {\sequent{\Gamma}{\judge{\Skip}{b \lor \phi}}\qquad
      \sequent{\Gamma,\, \judge{Y}{\phi}}
      {\judge{\Skip; S; Y}{\lnot b \lor \phi}}}
    {\sequent{\Gamma}{\judge{\While\ b\ \Do\ S}{\phi}}}
  \]
  where $\Gamma$ contains judgments of the form $\judge{Y}{\phi}$
  originating from while loops encountered previously. These are only
  needed in a proof in the presence of nested loops.
\end{remark}


\subsubsection{Semantics}

\begin{definition}[Calculus Semantics]
  \label{def:calculus-semantics}
  A judgment $\judge{S}{\phi}$ is termed \emph{valid}
  in~$\mathcal{I}$, denoted $\models_\mathcal{I} \judge{S}{\phi}$,
  whenever $\trcdenoti{S} \subseteq \fdenot{\phi}$.
  A sequent $\sequent{\Gamma}{\judge{S}{\phi}}$ is termed
  \emph{valid}, denoted $\semsequent{\Gamma}{\judge{S}{\phi}}$, if for
  every inter\-pretation~$\mathcal{I}$, $\judge{S}{\phi}$ is valid
  in~$\mathcal{I}$, whenever all judgments in~$\Gamma$ are valid
  in~$\mathcal{I}$.
\end{definition}

It is not possible to prove a judgment for $m()$ (or any other
statement) that is stronger than its strongest trace formula. In this
sense, $\stfs{m()}$ can be seen as a \emph{contract} for $m$, in fact
the \emph{strongest possible} contract. This is captured in the
following result:

\begin{corollary}[Strongest Trace Formula]
  \label{cor:stf-strongest}
  Let $S$ be a statement not involving any statement variables. Then
  the strongest trace formula $\stfs{S}$ of~$S$ entails any valid
  formula for~$S$.
  That is, if $\models \judge{S}{\phi}$, then $\stfs{S} \models \phi$.
\end{corollary}

\begin{proof}
  The result follows directly from
  Theorem~\ref{thm:stf-characterisation-rec} in the Appendix and
  Definition~\ref{def:calculus-semantics}.  
\end{proof}


\subsection{Soundness and Relative Completeness of the Calculus} 

Our proof system is \emph{sound}, in the sense that it can only
derive valid sequents.

\begin{theorem}[Soundness]
  \label{thm:calculus-soundness}
  The proof system is sound: every derivable sequent is valid.
\end{theorem}



Our proof system is \emph{complete}, in the sense that every valid
sequent can be derived, \emph{relative} to an oracle, used by rule
(\textsc{Cons}), that provides logical entailment between trace
formulas.

By Theorem~\ref{thm:stf-characterisation-rec} and 
Definition~\ref{def:calculus-semantics} we know that every judgment
of the shape $\judge{S}{\stfs{S}}$ is valid. 
We next show that all such judgments are derivable in our proof system.
Together with Corollary~\ref{cor:stf-strongest}, we obtain completeness,
with the help of the rule (\textsc{Cons}).
By definition:
$
  \stfs{m()}=\mathit{Id}\chop\mu X_m.\,\stfh{\{m\}}{S_m} 
$,
where $S_m$ is the body of $m$. We abbreviate
$\phi_m=\mu X_m.\,\stfh{\{m\}}{S_m}$ and in the following use
$\stfs{m()}=\mathit{Id}\chop\phi_m$ without mentioning it explicitly.

\begin{theorem}[Existence of Canonical Proof]
  \label{thm:existence-rec}
  Let $\langle S, T\rangle$ be a \Rec\ program with $n$ many method
  declarations in~$T$, and let
  $
    \Gamma=\{\judge{Y_{m_1}}{\phi_{m_1}},\ldots,\judge{Y_{m_n}}{\phi_{m_n}}\}.
  $
  Then, the judgment
  $
    \sequent{\Gamma}{\judge{S[\Skip;Y_{m_1}/m_1(),\ldots,\Skip;Y_{m_n}/m_n()]}{\stfs{S}}}
  $
  is derivable in our calculus. 
\end{theorem}

\begin{corollary}[Relative Completeness]
  \label{cor:relative-completeness-rec}
  The proof system is relatively complete: for every \Rec\ program $S$
  without statement variables and every closed formula $\phi$, any
  valid judgment of the form $\judge{S}{\phi}$ is derivable in the
  proof system.
\end{corollary}

\begin{proof}
  By Theorem~\ref{thm:stf-characterisation-rec} we know that
  $\stfs{S}\models\phi$, so we can use rule~(\textsc{Cons}) to obtain
  $\judge{S}{\stfs{S}}$, which is derivable by
  Theorem~\ref{thm:existence-rec}. 
\end{proof}

Compared to a typical completeness proof of first-order Dynamic Logic,
where the invariant is constructed as equations over the G\"{o}delized
program in the loop, the argument is much simpler, because the
inductive specification logic is sufficiently expressive to
characterize recursive programs (and loops as a special case).
First-order quantifiers are not even necessary, so our logic
is \emph{not} first-order, even though it is obviously Turing-hard and
thus undecidable.


\section{From Trace Formulas to Programs}
\label{sec:canonical-programs}

In Section~\ref{sec:stf} we showed that any \Rec\ program $S$ can be
translated into a trace formula $\stfs{S}$ that has the same semantics
in terms of traces. Now we look at the other direction: Given a trace
formula~$\phi$, can we construct a \emph{canonical} program
$\can{\phi}$ that has the same semantics in terms of traces?
In general, this is not possible, as the following example shows:

\begin{example}\label{ex:nondet}
  Consider the trace formula:
  $
    \mathit{Sb}_y^0\chop\mu X.\left(\mathit{Id}\lor\mathit{Sb}_y^{y+1}\chop X\right)
  $.
  Its semantics are the traces that count $y$ up from $0$ to any
  finite number. It is not possible to model the non-deterministic
  choice in the fixed point formula directly in \Rec, because the
  number of calls is unbounded.

  There is a \Rec\ program that produces exactly the same traces as
  the formula above, up to auxiliary variables, for example,
  $y:=0;m()$, where $m$ is declared as:
  $m()\,\{\If\ (y\leq x)\ \Then\ y:=y+1;m()\ \Else\ \Skip
  \}$. However, to transform an arbitrary formula with unbounded
  non-determinism in the number of calls to an equivalent one with
  non-deterministic initialization, is difficult and not natural.
\end{example}


\subsection{\Rec\ Programs with Non-deterministic Choice}
\label{sec:recast-programs}

To achieve a \emph{natural} translation from the trace logic to
canonical programs, it is easiest to introduce \emph{non-deterministic
  choice} in the form of a statement
$\If\ \ast\ \Then\ S_1\ \Else\ S_2$. The extension of language \Rec\
with the corresponding grammar rule is called $\Rec^\ast$.

The SOS rules for non-deterministic choice are:
\begin{center}
\Rule{$\ast$-$i$}
  {-}
  {\SOSConfTrans{\If\ \ast\ \Then\ S_1\ \Else\ S_2}{s}{S_i}{s}} \quad $i \in \set{1, 2}$ 
\end{center}

The finite-trace semantics of non-deterministic choice is:
\begin{center}
  $   \trcdenot{\If\ \ast\ \Then\ S_1\ \Else\ S_2}   \defeq  
      \sharp \trcdenot{S_1} \>\cup\>\> 
      \sharp \trcdenot{S_2}      $
\end{center}

The extension of Theorem~\ref{thm:semantics-correctness-rec} for
non-deterministic choice is completely straightforward.
Likewise, the theory of strongest trace formulas is easy to extend:
%
\begin{center}
  $\stfs{\If\ \ast\ \Then\ S_1\ \Else\ S_2} \defeq \mathit{Id} \chop
  \stfs{S_1} \lor \mathit{Id} \chop \stfs{S_2}$
\end{center}

It is easy to adapt the proof of
Theorem~\ref{thm:stf-characterisation-rec}.
The corresponding calculus rule is:
\begin{center}
  $
  \Rule{If-$\ast$}
  {\sequent{\Gamma}{\judge{\Skip; S_1}{\phi}} \quad
   \sequent{\Gamma}{\judge{\Skip; S_2}{\phi}}}
  {\sequent{\Gamma}{\judge{\If\ \ast\ \Then\ S_1\ \Else\ S_2}{\phi}}} 
  $
\end{center}

It is also easy to extend the proofs of
Theorem~\ref{thm:calculus-soundness} and
Theorem~\ref{thm:existence-rec}.
A more problematic aspect of the translation to canonical programs
concerns formulas of the form $\phi_1\land\phi_2$, because there is no
natural programming construct that computes the intersection of
traces. But in Definition~\ref{def:stf} general conjunction is not
required, so without affecting the results in previous sections we can
restrict the syntax of trace formulas in
Definition~\ref{def:logic-syntax} to $p\wedge\phi$.

A final issue are the programs that characterize a trace formula of
the form~$p$ with semantics $\fdenot{p}=\State^+\restr{p}{}$. A
program that produces such traces requires a \Any\ statement that
resets \emph{all} variables to an arbitrary value. This goes beyond
non-deterministic choice quite a bit, but luckily, it is not required:
As seen above, trace formulas of the form $p$ occur only as
subformulas of $p\wedge\phi$. Further, formulas of the form
$p\vee\phi$ occur only as intermediate formulas in derivations, and
nowhere else.
Altogether, for the purpose of mapping formulas to programs, we can
leave out the production for~$p$.
The grammar in Definition~\ref{def:logic-syntax} is thus simplified
to:
$$ \phi \>::=\> 
                \mathit{Id} \mid \mathit{Sb}^a_x \mid 
                X \mid 
                p \wedge \phi \mid 
                \phi \vee \psi \mid 
                \phi \chop \psi \mid 
                \mu X. \phi
$$


In addition, we assume without loss of generality that all recursion
variables in a trace formula have unique names.


\subsection{Canonical Programs}
\label{sec:canonical-programs-1}

\begin{definition}[Canonical Program]
\label{def:can}
  Let $\phi$ be trace formula. The \emph{canonical program}
  for~$\phi$, denoted $\can{\phi}=\langle S_\phi,T_\phi\rangle$, is
  inductively defined as follows:
  \begin{align*}
    \can{\mathit{Id}} & \defeq \langle\Skip,\epsilon\rangle &
    \can{p\wedge\phi} & \defeq \langle\If\ p\ \Then\ S_\phi\ \Else\ \Abort,T_\phi\rangle \\
    \can{\mathit{Sb}^a_x} & \defeq \langle x:=a,\epsilon\rangle &
    \can{\phi\vee\psi} & \defeq \langle\If\ \ast\ \Then\ S_\phi\ \Else\ S_\psi,T_\phi\, T_\psi\rangle\\
    \can{\phi\chop\psi} & \defeq \langle S_\phi;S_\psi,T_\phi\, T_\psi\rangle &
    \can{\mu X.\phi} & \defeq \langle m_X(),T_\phi\,\{m_X \{S_\phi\}\}\rangle\\
    \can{X} & \defeq \langle m_X(),\epsilon\rangle
  \end{align*}
\end{definition}

The definition contains the statement
\Abort. 
It is definable in the \Rec\ language as
$\Abort \defeq \mathit{abort()}$, with the declaration
$\mathit{abort}\,\{\mathit{abort()}\}$.
We assume that any table $T_\phi$ contains the declaration of procedure
$\mathit{abort}()$ when needed.

\begin{example}
We translate the formula in Example~\ref{ex:nondet}: 
$$ \can{\mathit{Sb}_y^0\chop\mu X.(\mathit{Id}\lor\mathit{Sb}_y^{y+1}\chop X)} =
   \langle y:=0;m_X(),\,T\rangle $$
where $T=m_X\,\{\If\ \ast\ \Then\ \Skip\ \Else\ y:=y+1;m_X()\}$.
\end{example}


\begin{proposition}\label{prop:any}
  We have $\trcdenot{\Abort}=\varnothing$.
\end{proposition}

\begin{proposition}\label{prop:well}
  Let $\phi$ be an open trace formula, let $T_\phi'$ be declarations
  of its unbound recursion variables, and let
  $\can{\phi}=\langle S_\phi,T_\phi\rangle$. Then
  $\langle S_\phi,T_\phi\,T_\phi'\rangle$ is a well-defined
  $\Rec^\ast$ program.
\end{proposition}



Evaluation of procedure calls and Boolean guards introduce \emph{stuttering 
steps} as compared to the corresponding logical operators of least
fixed-point recursion and disjunction, respectively. Hence, canonical
programs obtained from a formula~$\phi$ are statements, whose trace
semantics is equal to the one of~$\phi$, but \emph{modulo stuttering}: 
The two trace sets are equal when abstracting away the stuttering steps.
Further, we say that statement~$S_1$ \emph{refines} statement~$S_2$,
written $S_1 \preceq S_2$, when $\trcdenot{S_1} \subseteq \trcdenot{S_2}$.

\begin{definition}[Stuttering Equivalence]
  Let $\tilde{\sigma}$ be the \emph{stutter-free} version of a trace
  $\sigma$, i.e., where any subtrace of the form $s\cdot s\cdots s$
  has been replaced with $s$. Define
  $\widetilde{A}=\{\tilde{\sigma}\mid\sigma\in A\}$.  We say two trace
  sets $A$, $B$ are \emph{stutter-equivalent}, written
  $A\,\widetilde{=}\,B$, if $\widetilde{A}=\widetilde{B}$.
\end{definition}

It is easy to see that $A=B$ implies $A\,\widetilde{=}\,B$.
Let $\widetilde{\models}$ and $\widetilde{\preceq}$ 
denote entailment between formulas and refinement between
statements, respectively, both modulo stuttering equivalence.

Unsurprisingly, the characterization of canonical programs resembles
the one of strongest trace formulas, however, modulo stuttering
equivalence.

\begin{theorem}[Characterisation of Canonical Program]
\label{thm:can-characterisation-rec}
  Let $\phi$ be a closed trace formula, and let
  $\can{\phi}=\langle S_\phi,T_\phi\rangle$.
  %
  Then
  $\trcdenot{S_\phi} \,\widetilde{=}\, \fdenot{\phi}$.
\end{theorem}

Finally, we can establish that $\stfs{\cdot}$ and $\can{\cdot}$ form
a Galois connection w.r.t.\ the partial orders $\widetilde{\models}$ 
on formulas and $\widetilde{\preceq}$ on statements.

\begin{corollary}
\label{cor:can}
  Let $\phi$ be a closed trace formula, and let
  $\can{\phi}=\langle S_\phi,T_\phi\rangle$. Then,
  for every statement~$S$, we have:
  $ \stfs{S} \,\widetilde{\models}\, \phi
     \quad\text{iff}\quad
     S \,\widetilde{\preceq}\, S_\phi $.
\end{corollary}

\begin{proof}
  By using Theorem~\ref{thm:stf-characterisation-rec}
  and Theorem~\ref{thm:can-characterisation-rec}. 
\end{proof}


\section{Related Work}
\label{sec:related-work}

We do not discuss higher-order logical frameworks
\cite{Coq04,Isabelle02}. Even though these are expressive, in program
verification they are not used to specify trace-based properties of
programs, but rather to mechanize conventional contract-based
deductive verification \cite{Oheimb01}.

Stirling~\cite[p.~528, footnote~2]{sti-93-handbook-chapter} suggests
that the $\mu$-calculus can be generalized to non-unary predicates,
but does not develop this possibility further. 
In~\cite{olm-99}, Müller-Olm proposes a modal fixed-point logic with chop,
which can characterize any context-free process up to bisimulation or 
simulation. The logic is shown to be strictly more expressive than the 
modal $\mu$-calculus.
Lange \& Somla \cite{LangeSomla06} relate propositional dynamic logic
over context-free programs with Müller-Olm's logic and show the former
to be equi-expressive with a fragment of the latter.
Fredlund et al.~\cite{ACDFGN03} presented a verification tool for the
\textsc{Erlang} language based on first-order $\mu$-calculus with
actions~\cite{DamGurov02}.

In contrast to these papers, we separate programs from fixed-point
formulas and relate them in the form of judgments. Our logic has only
a single binary operator $\mathit{Sb}_x^{a}$ over arithmetic
expressions~$a$ and program variables~$x$, together with the chop
operator $\chop$.  The latter models composition of binary relations
in the denotational semantics. Our logic is sufficient to characterize
any $\Rec$ program. Specifically, $\mu$-formulas can serve as
contracts of recursive procedures.  More importantly, our approach
leads to a \emph{compositional} calculus, where all rules but the
consequence rule are analytic.

Kleene algebra with tests (KAT) \cite{Kozen97} are an equational
algebraic theory that has been shown to be as expressive as
propositional while programs. They have been mechanized in an
interactive theorem prover \cite{Pous13} and are able to express at
least Hoare-style judgments~\cite{Kozen00}. The research around
focuses around \emph{propositional while} programs: we are not aware
of results that relate KAT with recursive stateful
programs. Specifically, our result that procedure contracts can
expressed purely in terms of trace formulas
(Theorem~\ref{thm:stf-characterisation-rec}) has not been obtained by
algebraic approaches.

Expressive trace-based specification languages are relatively rare in
program verification.
The trace logic of Barthe et al.~\cite{BEGGKM19} actually is a
many-sorted \emph{first-order logic}, equipped with an arithmetic
theory of explicit trace positions to define program semantics. It is
intended to model program verification in first-order logic for
processing in automated theorem provers. Like $\trcdenot{}$, their
program semantics is compositional; however, it uses explicit time
points instead of algebraic operators.
An extension of Hoare logic with trace specifications is presented by
Ernst et al.~\cite{ErnstKM22}.  The standard Hoare-style pre- and
postcondition for a statement is extended with regular expressions
recording events emitted before the execution of the command and the
events emitted by its execution. Our trace logic is more expressive.
Also the first-order \emph{temporal logic of nested words} (NWTL) of
Alur et al.~\cite{AABEIL08} permits to specify certain execution
patterns within symbolic traces. It is orthogonal to our approach,
being based on nested event pairs and temporal operators, instead of
least fixed points and chop. NWTL is equally expressive over nested
words as first-order logic. The intended computational model is not
\Rec, but the more abstract non-deterministic Büchi automata over
nested words for which it is complete.
The \emph{temporal stream logic} of Finkbeiner et al.~\cite{FKPS19},
like our trace logic, has state updates $\mathit{Sb}_x^t$ (with a
different syntax) and state predicates, but it is based on linear
temporal logic and has no least fixed point or chop operator. Again,
the intended computational model is an extension of Büchi automata,
the verification target are FPGA programs.

Cousot \& Cousot \cite{CousotCousot00} define a trace-based semantics
for modal logics where (infinite) traces are equipped with past,
present, and future. Their main focus is to relate model checking and
static analysis to abstract interpretation---the trace-based semantics
is the basis for it. In contrast, our paper relates a computation
model to a logic. Like our semantics, theirs is compositional and the
operators mentioned in their paper could inspire abstractions of our
trace logic, cf.~Section~\ref{sec:abstr-spec} below.

Nakata \& Uustalu~\cite{nak-uus-09} present a trace-based co-inductive
operational semantics with chop for an imperative programming language
with loops. Following up on this work, \cite{DHJPT17} extended the
approach to an asynchronous concurrent language, but neither of these
uses fixed-points, so that the specification language is
incomplete. Also the calculus is not compositional.

Closest to our work is trace-based deductive
verification~\cite{bub-et-al-23}, using a similar trace logic as in
the present paper, but neither their semantics nor the calculus are
compositional. Also they prove soundness, but not completeness which
is left as an open question.


\section{Future Work}
\label{sec:future-work}

\subsection{Abstract Specifications and Extension of Recursive
  Programs}
\label{sec:abstr-spec}

It is desirable to formulate specifications in a more abstract manner
than the programs whose behavior they are intended to capture.  For
example, if we reinstate the atomic trace formula $p$ in our logic, we
can easily express the set of all finite traces as
$\fdenot{\mathit{true}}=\State^+$. Then we can define a binary
connective as
$\phi\cdot\cdot\,\psi\defeq\phi\chop\mathit{true}\chop\psi$ as ``any
finite (possibly, empty) computation may occur between $\phi$ and
$\psi$''. For example, the trace formula $\mu X.(X \cdot\cdot\, X)$
expresses that a procedure $m_X$ calls itself at least twice
recursively, at the beginning and at the end of its body,
respectively.

A more general approach to introduce non-determinism to the logic is
to define for a Boolean expression $b$ a \emph{binary relation} $R_b$
with semantics $R_b(s,s')\defequiv\bexdenot{b}(s)=\ttt$. The atomic
formula $p$ is then \emph{definable} as:
$p\defeq R_p\chop R_\mathit{true}^+$.
The advantage of basing $p$ on a binary relation is that it is more
easily represented as a canonical program. To this end, we introduce
an atomic statement $\Any$ with the trace semantics
$\trcdenot{\Any}=\{s\cdot s'\mid s,s'\in\State\}$, i.e.~in any given
state $s$, executing $\Any$ results in an arbitrary successor state.
The proof rule for $\Any$ is the axiom
$\sequent{\Gamma}{\judge{\Any}{R_\mathit{true}}}$ and obviously
$\stfs{\Any}\defeq R_{\mathit{true}}$.  Then $p$ is characterized by
$\can{p}\defeq\langle\If\ p\ \Then\ \mathit{havoc()}\ \Else\
\Abort,T\rangle$, where $T$ contains the declaration
$
\mathit{havoc}\ \{\If\ \ast\ \Then\ \Any\ \Else\ \Skip;\mathit{havoc()}\}.
$
  
Interestingly, $R_\mathit{true}$ (and, therefore, $\Any$ on the side
of programs) permits to define \emph{concatenation} of trace formulas
$\phi\cdot\psi$ with the obvious semantics:
\[
  \fdenot{\phi\cdot\psi}\,\defeq\,\{s_0\cdots s_n\cdot s'_0\cdots s'_m\mid s_0\cdots s_n\in\fdenot{\phi},\,s'_0,\cdots s'_m\in\fdenot{\psi}\}\,=\,\fdenot{\phi\chop R_{\mathit{true}}\chop\psi}
\]

\subsection{Proving Consequence of Trace Formulas}

In general this is a difficult problem that requires fixed-point
induction, but the derivations needed in practice might be relatively
simple, as the following example shows.

\begin{example}\label{ex:con}
  Consider the two trace formulas: 
  \begin{align*}
    \stfs{\mathit{down()}} & =\,\mathit{Id}\chop\mu X_{\mathit{down}}.
    \left((x>0\land\mathit{Id}\chop\mathit{Sb}_x^{x-2}\chop\mathit{Id}\chop X_{\mathit{down}})\lor(x\leq0\land\mathit{Id\chop}\mathit{Id})\right)\\
    \mathit{Dec}_x^+ & =\,\mu X_{\mathit{dec}}.(\mathit{Dec}_x\chop X_{\mathit{dec}}\lor\mathit{Dec}_x)
  \end{align*}
  from Examples~\ref{ex:rec-program-simple}
  and~\ref{ex:stf-rec-simple}, respectively. We expect the trace
  formula implication
  $\stfs{\mathit{down()}}\Rightarrow\mathit{Id}\chop\mathit{Dec}_x^+$
  to be provable, because of
  Theorem~\ref{thm:stf-characterisation-rec}.
\end{example}

It turns out that the following fixed-point induction rule and
consequence rule, combined with straightforward \emph{first-order}
consequence and logic rules, are sufficient to prove the claim:
\[
  \Rule{FP-Ind}
  {\sequent{\Gamma}{\phi\Rightarrow\psi}}
  {\sequent{\Gamma}{\mu X. \phi\Rightarrow\mu X. \psi}}
  \qquad
  \Rule{Cons-Left}
  {\sequent{\Gamma,\,\phi}{\psi}\qquad\sequent{\Gamma,\,\phi'}{\psi'}}
  {\sequent{\Gamma,\,\phi\chop\phi'}{\psi\chop\psi'}}
\]

In fact, the proof is considerably shorter than proving the judgment
$\judge{\mathit{down}()}{\mathit{Id}\chop\mathit{Dec}_x^+}$ in the
calculus of Section~\ref{sec:proof-calculus}, which is as well possible.

A proof system for trace formula implication that can prove the above
as well as many other non-trivial examples is given
in~\cite{Heidler24}.

\subsection{Non-terminating Programs}
\label{sec:non-term-progr}

Our results so far are limited to \emph{terminating} programs, i.e.~to
sets of finite traces. To extend the calculus with a termination
measure, such that a proof of $\judge{S}{\phi}$ not only shows
correctness of $S$ relative to $\phi$, but also ensures it produces
only finite traces, is easy.

However, the trace-based setup permits, in principle, also to prove
properties of \emph{non-terminating} programs. To this end, it is
necessary to extend the logic with operators whose semantics contains
infinite traces. One obvious candidate are \emph{greatest
  fixed-points}~\cite{sti-93-handbook-chapter}. The downside to this
approach is that nested fixed points of opposite polarity are
difficult to understand, as is well known from $\mu$-calculus. Is
there a restriction of mixed fixed-point formulas that
\emph{naturally} corresponds to a certain class of programs? The
theory of strongest trace formulas and canonical programs might guide
the search for such fragments.


\section{Conclusion}
\label{sec:extro}

We presented a fixed-point logic that characterizes recursive programs
with non-deterministic guards. Both, programs and formulas have the
same kind of trace semantics, which, like the calculus for proving
judgments, is fully compositional in the sense that the definitions
and rules embody no context.
The faithful embedding of programs into a logic seems to suggest that
we merely replace one execution model (programs) with another (trace
formulas). So why is it worth having such an expressive specification
logic? We can see four reasons:


First, the logic renders itself naturally to extension and abstraction
that cannot be easily mimicked by programs or that are much less
natural for programs.
This is corroborated by the discussion in
Section~\ref{sec:abstr-spec}, but also by the case of conjunction: It
is trivial to add conjunction $\phi\land\psi$ to the trace logic and
to a calculus for trace formulas, but conjunction has no natural
program counter-part. Yet it permits to specify certain
hyper-properties, i.e., properties relating sets of traces.

Second, the logic offers reasoning patterns that are easily justified
algebraically, such as projection, replacement of equivalents,
strengthening, distribution, etc., that are not obvious in the realm
of programs.

Third, the concept of \emph{strongest trace formula} leads to a
characterization of valid judgments and, thereby, enables a simple
completeness proof.

And finally, the duality between programs and formulas permits to
prove judgments by freely mixing two styles of reasoning: with the
rules of the calculus in Figure~\ref{fig:calculus}, or using a calculus
for the consequence of trace formulas.
For example, a judgment such as
$\judge{\mathit{down()}}{\mathit{Id}\chop\mathit{Dec}_x^+}$ can be
proved as in Example~\ref{ex:con} or directly with the rules in
Figure~\ref{fig:calculus}, but also by \emph{mixing} both styles.

A perhaps surprising feature of our trace logic is the fact that no
explicit notion of \emph{procedure contract} is required to achieve
procedure-modular verification: Instead, strongest trace formulas and
statement variables are employed.
This results in a novel (\textsc{Call}) rule that works with
\emph{symbolic continuations} realized by statement variables.




\bibliography{tracesem,reiner}


\appendix

\section{Proofs and More Examples}

\subsection{Semantics}
\label{sec:semantics-app}

\begin{example}\label{ex:sos-simple}
  Let us execute the statement $\Skip ; x := x-1$ (with empty
  procedure table) from some arbitrary initial state~$s$ in the above
  SOS.
  That is, let us apply the rules to derive all configurations
  reachable from the initial one, which is
  $\config{\Skip ; x := x-1}{s}$.
  First, by means of the rules \textsc{Skip} and \textsc{Seq-1}, we
  derive the transition
  $\SOSConfTrans{\Skip ; x := x-1}{s}{x := x-1}{s}$.
  Then, by applying the \textsc{Assign} rule, we derive the transition
  $\SOSConfFinTrans{x := x-1}{s}{s [x \mapsto s(x)-1]}$.
  We reached a final configuration, and thus execution terminates.
  The program execution (or run) we obtain is:
  $
  \config{\Skip ; x := x-1}{s}  \Rightarrow
  \config{x := x-1}{s}          \Rightarrow
  s [x \mapsto s(x)-1]
  $
\end{example}

\begin{example}\label{ex:sos-rec}
  We execute the statement $x:=2;\mathit{down}()$ with the procedure
  table from Example~\ref{ex:rec-program-simple}.  The program
  execution (or run) we obtain from an arbitrary state $s$ is:
  $$
  \begin{array}{l}
  \config{x:=2;\mathit{down}()}{s}  \Rightarrow
  \config{\mathit{down}()}{s[x \mapsto 2]}  \Rightarrow \\
  \config{\If\ x > 0\ \Then\ x := x-2; \mathit{down}()\ \Else\ \Skip}{s[x \mapsto 2]}  \Rightarrow \\
  \config{x := x-2; \mathit{down}()}{s[x \mapsto 2]}  \Rightarrow 
  \config{\mathit{down}()}{s[x \mapsto 0]}  \Rightarrow \\
  \config{\If\ x > 0\ \Then\ x := x-2; \mathit{down}()\ \Else\ \Skip}{s[x \mapsto 0]}  \Rightarrow 
  \config{\Skip}{s[x \mapsto 0]}  \Rightarrow s[x \mapsto 0]
  \end{array}
  $$
\end{example}

\begin{example}
  Building on Example~\ref{ex:sos-simple}, the induced finite-trace
  semantics of the statement $\Skip ; x := x-1$ is the set
  $
    \setdef{s \cdot s \cdot s[x \mapsto s(x) - 1]}{s \in \State}
  $
  of finite traces, all of which are of length~$3$.
\end{example}

We state some useful algebraic \emph{laws} for finite-trace sets which
we use later in our proofs:
$$
\begin{array}{r@{\;}c@{\;}l@{\qquad}r@{\;}c@{\;}l@{\qquad}r@{\;}c@{\;}l}
  A \chop (B \cup C) & = & A \chop B \>\cup\, A \chop C &
   (A \cup B) \chop C & = & A \chop C \>\cup\, B \chop C \\
   A \subseteq B & \Rightarrow & A \chop C \subseteq B \chop C & 
   (A \chop B) \restr{b}{} & = & A \restr{b}{} \!\!\chop B  &
   \sharp (A \chop B) & = & \sharp A \chop B \\     
\end{array}
$$

The following ``unfolding'' equivalence for procedure calls holds:

\begin{proposition}
\label{prop:call-unfolding-2}
Let $m\,\{S\}$ be a method declaration. Then the following semantic
equality holds:
$
  \trcdenot{m()} \,=\, \sharp \trcdenot{S} 
$
\end{proposition}

\begin{proof}
We have:
$$ \begin{array}{r@{\;}c@{\;}ll}
   \trcdenot{m()} & \>=\> & \trcdenoth{m()}_{\rho_0} &~~~~~
      \{\mbox{By the definition of $\trcdenot{S}$}\} \\
   & = & \rho_0 (m) &~~~~~
      \{\mbox{By the definition of $\trcdenoth{S}_{\rho}$}\} \\
   & = & \sharp \trcdenoth{S}_{\rho_0} &~~~~~
      \{\mbox{Since $\rho_0$ is a fixed-point of $H$}\} \\
   & = & \sharp \trcdenot{S} &~~~~~
      \{\mbox{By the definition of $\trcdenot{S}$}\} \\
   \end{array} $$
\end{proof}



\noindent\textbf{Fixed-Point Induction and Beki\v{c}'s Principle.}
It can be shown that
$\rho_0 = \bigsqcup_{k \geq 0} H^k (\varnothing^n)$, where $\sqcup$ is
point-wise set union, a result originally due to Kleene, but popularly
known as the Knaster-Tarski Fixed-Point Theorem.
The interpretations $\rho^k \defeq H^k (\varnothing^n)$ are referred
to as \emph{fixed-point approximants} of~$H$.

Then, for a given transformer~$H$ and set~$A$, a result of the type
$\mathit{LFP}\, H \sqsubseteq A$ can be established by proving
$H^k (\varnothing) \sqsubseteq A$ for all~$k$, by mathematical
induction on~$k$. Since $\bigsqcup_{k \geq 0} H^k (\varnothing^n)$ is
the least upper bound of the set of approximants, we obtain
$\mathit{LFP}\, H \sqsubseteq A$.
But sometimes one can prove a stronger result, namely that
$\gamma \sqsubseteq A$ implies $H (\gamma) \sqsubseteq A$ for
all~$\gamma$. This immediately entails the above proof by mathematical
induction---a proof schema commonly referred to as \emph{Fixed-Point
  Induction}.
We make use of these proof principles below.

We will also make use of another principle, known as \emph{Beki\v{c}'s
  Principle}, which allows a simultaneous fixed-point, such
as~$\rho_0$, to be expressed as a series of individual least
fixed-points~\cite{bek-69}.
Using this principle, we can apply the above Fixed-Point Induction
Principle, but on one procedure name at a time.


\smallskip

\noindent\textbf{Presentation of Proofs.}
While the following inductive proofs over the statement structure were
done for all cases, for sake of conciseness, we reproduce only the
cases for conditional and method call: All other cases are
straightforward and carry no additional information.


\begin{proof}[Proof of Theorem~\ref{thm:semantics-correctness-rec}]
  
Typically, such results are proved separately in the two directions,
utilizing different induction principles.

\noindent\textbf{(1) Proof of $\trcdenot{S} \subseteq \sosdenot{S}$.}

For $\sigma \in \State^+$,
$\sigma \in \trcdenot{S}$ is equivalent by definition to
$\sigma \in \trcdenoth{S}_{\rho_0}$, which in turn is equivalent to
$\exists k \geq 0.\, \sigma \in \trcdenoth{S}_{\rho^k}$.
Hence, what we want to show is logically equivalent to showing that 
$\trcdenoth{S}_{\rho^k} \subseteq \sosdenot{S}$ for all $k \geq 0$. 
We show this by mathematical induction on~$k$.

The base case of $k = 0$ holds vacuously.
For the induction case, assume 
$\trcdenoth{S}_{\rho^k} \subseteq \sosdenot{S}$ for an
arbitrary $k \geq 0$.
We proceed by (an inner) induction on the structure of~$S$.
\emph{We proved all cases, but for lack of space, for this and all
  results that follow, we reproduce only the most interesting cases.}

\noindent\textbf{Case $S = \If\ b\ \Then\ S_1\ \Else\ S_2$.}
Assume the result holds for~$S_1$ and~$S_2$.
Let 
$s_0 \cdot s_1 \cdot \ldots \cdot s_n \in 
 \trcdenot{\If\ b\ \Then\ S_1\ \Else\ S_2}$.
Then, by Definition~\ref{def:while-tr}, 
$s_0 = s_1$, and either $\bexdenot{b} (s_0) = \ttt$ and 
$s_1 \cdot \ldots \cdot s_n \in \trcdenot{S_1}$,
or else $\bexdenot{b} (s_0) = \fff$ and 
$s_1 \cdot \ldots \cdot s_n \in \trcdenot{S_2}$.
By the induction hypothesis, then 
either $\bexdenot{b} (s_0) = \ttt$ and 
$s_1 \cdot \ldots \cdot s_n \in \sosdenot{S_1}$,
or else $\bexdenot{b} (s_0) = \fff$ and 
$s_1 \cdot \ldots \cdot s_n \in \sosdenot{S_2}$.
By Definition~\ref{def:induced-finite-trace-semantics},
and by applying rule \textsc{If-1} or \textsc{If-2},
respectively, we obtain that
$s_0 \cdot s_1 \cdot \ldots \cdot s_n \in 
 \sosdenot{\If\ b\ \Then\ S_1\ \Else\ S_2}$.

\noindent\textbf{Case $S = m ()$} Let~$m$ be declared as $m\, \{S'\}$.
Let 
$s_0 \cdot s_1 \cdot \ldots \cdot s_n \in \trcdenoth{m ()}_{\rho^{k+1}}$.
Then, by the definition of $\trcdenot{S}$, we have
$s_0 \cdot s_1 \cdot \ldots \cdot s_n \in \rho^{k+1} (m)$.  But
$\rho^{k+1} = H (\rho^k)$, and therefore
$s_0 \cdot s_1 \cdot \ldots \cdot s_n \in \sharp
\trcdenoth{S'}_{\rho^k}$, by definition of $H$.
Hence, $s_0 = s_1$ and
$s_1 \cdot s_2 \cdot \ldots \cdot s_n \in \trcdenoth{S'}_{\rho^k}$.
By the outer induction hypothesis,
$s_1 \cdot s_2 \cdot \ldots \cdot s_n \in \sosdenot{S'}$, hence
$s_0 \cdot s_1 \cdot \ldots \cdot s_n \in \sharp \sosdenot{S'}$. By
Rule~\textsc{Call} and
Definition~\ref{def:induced-finite-trace-semantics}, we finally obtain
that $s_0 \cdot s_1 \cdot \ldots \cdot s_n \in \sosdenot{m ()}$.

\noindent\textbf{(2) Proof of $\sosdenot{S} \subseteq \trcdenot{S}$.}

The proof proceeds by (strong) induction on the length of traces.
Assume the result holds for all traces of length up to~$n$, for all
statements~$S$. Let $S$ be a statement and
$s_0 \cdot \ldots \cdot s_n \in \sosdenot{S}$ be a trace of
length~$n+1$. Then, by
Definition~\ref{def:induced-finite-trace-semantics}, there are
statements $S_0, S_1, \ldots, S_{n-1}$ such that $S_0 = S$,
$\SOSConfTrans{S_i}{s_i}{S_{i+1}}{s_{i+1}}$ for all
$0 \leq i \leq n-2$, and $\SOSConfFinTrans{S_{n-1}}{s_{n-1}}{s_n}$.
The proof that $s_0 \cdot \ldots \cdot s_n \in \trcdenot{S}$ proceeds
by case analysis on the SOS rule applied to obtain the first
transition of the SOS trace.
Again, we only show the most interesting cases.

\noindent\textbf{Rule} \textsc{If-1}. 
Then $S = \If\ b\ \Then\ S_1\ \Else\ S_2$, $s_0 = s_1$ and
$\bexdenot{b} (s_0) = \ttt$.
By Definition~\ref{def:induced-finite-trace-semantics},
$s_1 \cdot \ldots \cdot s_n \in \sosdenot{S_1}$,
and by the induction hypothesis, 
$s_1 \cdot \ldots \cdot s_n \in \trcdenot{S_1}$.
Since $s_0 = s_1$ and $\bexdenot{b} (s_0) = \ttt$, we have
$s_0 \cdot s_1 \cdot \ldots \cdot s_n \in 
 (\sharp \trcdenot{S_1}) \!\!\mid_b$.
Then, by Definition~\ref{def:while-tr},
$s_0 \cdot s_1 \cdot \ldots \cdot s_n \in 
 \trcdenot{\If\ b\ \Then\ S_1\ \Else\ S_2}=\trcdenot{S}$.

\noindent\textbf{Rule} \textsc{Call}. 
Then $S = m ()$ for a procedure~$m$ declared as $m\, \{S'\}$ for some
statement~$S'$, and $s_0 = s_1$.
By Definition~\ref{def:induced-finite-trace-semantics}, 
$s_1 \cdot \ldots \cdot s_n \in \sosdenot{S'}$, and then, by the
induction hypothesis, $s_1 \cdot \ldots \cdot s_n \in \trcdenot{S'}$.
Now, since $s_0 = s_1$, 
$s_0 \cdot s_1 \cdot \ldots \cdot s_n \in \sharp \trcdenot{S'}$ 
and hence, by Proposition~\ref{prop:call-unfolding-2},
$s_0 \cdot s_1 \cdot \ldots \cdot s_n \in \trcdenot{m ()}=\trcdenot{S}$. 

\end{proof}

\noindent\textbf{Fixed-Point Induction.}
As in the domain of finite traces,
one can also apply
fixed-point induction in proofs about recursive formulas.
In particular, a result of shape $\logdenot{\mu X. \phi} \subseteq A$
can be established by proving that $\gamma \subseteq A$ entails
$\logdenotsubst{\phi}{X}{\gamma} \subseteq A$.

\begin{proposition}[Restricted Traces and State Formulas]
  \label{prop:state-restr}
  Let $b$ be a Boolean expression, and $\phi$ a trace formula. Then
  $
    (\sharp\fdenot{\phi})\!\!\mid_b\,=\,\fdenot{b\land\mathit{Id}\chop\phi}
  $.
\end{proposition}

\begin{proof}
We have:
  \[
    \begin{array}{cl}
    & (\sharp\fdenot{\phi})\!\!\mid_b\\
    = & \setdef{s \cdot \sigma \in \sharp\fdenot{\phi}}{\bexdenot{b} (s) = \ttt} \\
    = & \setdef{s \cdot \sigma }{\bexdenot{b} (s) = \ttt} \cap\setdef{s \cdot s \cdot \sigma }{s\cdot\sigma\in\fdenot{\phi}}\\
    = & \setdef{s \cdot \sigma }{s\models b} \cap\setdef{s \cdot s}{s\in\State}\chop\setdef{s \cdot \sigma }{s\cdot\sigma\in\fdenot{\phi}}\\
    = & \fdenot{b} \cap(\fdenot{\mathit{Id}}\chop\fdenot{\phi})\\
    = & \fdenot{b\land\mathit{Id}\chop\phi}
    \end{array}
  \]
\end{proof}

\subsection{Strongest Trace Formula}
\label{sec:stf-app}

The proof of Theorem~\ref{thm:stf-characterisation-rec}---the
semantics of the strongest trace formula of a program turns is
identical to its trace semantics---requires an inner fixed-point
induction for which it is advantageous to break down $\mathsf{stf}$
into a \emph{modal equation system} $\mathsf{mes}$ \cite{Mader97} that
preserves the structure of procedure declarations.

\begin{definition}[Modal Equation System]
\label{def:meq}
  Given a closed trace formula $\phi$, the \emph{modal equation
  system} $\mes{\phi}$ is an open trace formula together with a set
  of modal equations of the form $X_i=\phi_i$, inductively defined
  over $\phi$ as follows: $\mes{\phi}$ is just the homomorphism over
  the abstract syntax, except when $\phi=\mu X.\phi'$: In this case,
  $\mes{\phi} \>\defeq\> X$, and a new equation $X=\mes{\phi'}$ is
  added.
\end{definition}

The semantics of modal equation systems is defined as in the
literature~\cite{Mader97}, from where we also take the semantic
equivalence between~$\phi$ and $\mes{\phi}$.
%

\begin{example}
  The modal equation system corresponding to the strongest trace
  formula in Example~\ref{ex:stf-rec} is: $\mes{\stfs{\mathit{even}()}} = \mathit{Id}\chop X_{\mathit{even}}$, where
  \begin{align*} 
    X_{\mathit{even}}  = \bigl((x = 0\, \land\, & \mathit{Id} \chop \mathit{Sb}^{1}_{y}) \lor (x \neq 0 \land \mathit{Id} \chop \mathit{Sb}^{x-1}_{x} \chop \mathit{Id} \chop X_{\mathit{odd}})\bigr)\\
    X_{\mathit{odd}}  = \bigl((x = 0\, \land\, & \mathit{Id} \chop \mathit{Sb}^{0}_{y}) \lor (x \neq 0 \land \mathit{Id} \chop \mathit{Sb}^{x-1}_{x} \chop \mathit{Id} \chop X_{\mathit{even}})\bigr)
  \end{align*}
  Observe the structural similarity between the formula
  $\mathit{Id}\chop X_{\mathit{even}}$ in the context of the defining
  equations for~$X_{\mathit{even}}$ and~$X_{\mathit{odd}}$, and the
  statement $\mathit{even}()$ in the context of the procedure
  table~$T$ of Example~\ref{ex:rec-program}.
\end{example}

To use the decomposition of a trace formula into a modal equation
system, we need to define for each program~$S$ an open trace formula
corresponding to $\mes{\stfs{S}}$.  This transformation, called $\stfsh{S}$,
is defined exactly as $\stfh{X}{S}$ in Definition~\ref{def:stf}
(ignoring the parameter~$X$), except for the case $S=m()$, which is
defined as $\stfsh{m()} = X_m$.

\begin{lemma}
\label{lem:stfsh-characterisation}
  Let $\langle S,T\rangle$ be a \Rec\ program and $M$ the procedures
  declared in $T$. Let $\rho : M \rightarrow 2^{\State^+}$ be an
  arbitrary interpretation of the procedures $m\in M$, and let
  $\mathcal{V}_\rho : \mathsf{RVar} \rightarrow 2^{\State^+}$ be the
  (induced) valuation of the recursion variables defined by
  $\mathcal{V}_\rho (X_m) \defeq \rho (m)$.
  We have, for all statements~$S$ of \Rec:
  $
  \logdenotV{\stfsh{S}}{\mathcal{V}_\rho}  \>=\>  \trcdenoth{S}_\rho
  $
\end{lemma}

\begin{proof}
The proof proceeds by induction on the structure of~$S$. 
%

\noindent\textbf{Case $S = \If\ b\ \Then\ S_1\ \Else\ S_2$.}
By the induction hypothesis,
$\logdenotV{\stfs{S_i}}{\mathcal{V}_\rho} = \trcdenoth{S_i}_\rho$ for $i=1,2$.
Using the obvious generalization of Proposition~\ref{prop:state-restr}
to open formulas, we have:  
\[
\begin{array}{cl}
    & \logdenotV{\stfsh{\If\ b\ \Then\ S_1\ \Else\ S_2}}{\mathcal{V}_\rho} \\
  = & \logdenotV{(b \land \mathit{Id} \chop \stfsh{S_1})
              \lor
              (\lnot b \land \mathit{Id} \chop \stfsh{S_2})}{\mathcal{V}_\rho} \\
  = & \logdenotV{(b \land \mathit{Id} \chop \stfsh{S_1})}{\mathcal{V}_\rho}
      \cup
      \logdenotV{(\lnot b \land \mathit{Id} \chop \stfsh{S_2})}{\mathcal{V}_\rho} \\
  = & (\sharp \logdenotV{\stfsh{S_1}}{\mathcal{V}_\rho}) \!\!\mid_b 
      \cup\>
      (\sharp \logdenotV{\stfsh{S_2}}{\mathcal{V}_\rho}) \!\!\mid_{\neg b}  \\
  = & (\sharp \trcdenoth{S_1}_\rho) \!\!\mid_b 
      \cup\>
      (\sharp \trcdenoth{S_2}_\rho) \!\!\mid_{\neg b}  \\
  = & \trcdenoth{\If\ b\ \Then\ S_1\ \Else\ S_2}_\rho  
\end{array}
\]

\noindent\textbf{Case $S = m()$.}
We have:
\[
\begin{array}{cl}
    & \logdenotV{\stfsh{m()}}{\mathcal{V}_\rho} \\
  = & \logdenotV{X_m}{\mathcal{V}_\rho}  \\
  = & \mathcal{V}_\rho (X_m)  \\
  = & \rho (m)  \\
  = & \trcdenoth{m()}_\rho  
\end{array}
\]
\end{proof}

\begin{proof}[Proof of Theorem~\ref{thm:stf-characterisation-rec}]
  Since $\stfsh{S}$ and $\stfh{X}{S}$ are defined identically, except
  for the case $S=m()$, the proof is the same as for
  Lemma~\ref{lem:stfsh-characterisation}, except for that case:
  
  \noindent\textbf{Case $S = m()$.} 
  %
  We only sketch the proof here. 
  We translate the formula $\stfh{\varnothing}{m()}$ into a modal
  equation system $\mes{\stfh{\varnothing}{m()}}$.
  This results in the formula~$X_m$, defined in the context of a
  system of modal equations: for each fixed-point operator $\mu X_i$
  in $\stfh{\varnothing}{m()}$, there is an equation
  $X_i = \mathit{Id} \chop \stfsh{S_i}$, whenever $m_i$ is declared as
  $m_i\, \{S_i\}$ in~$T$.
  Next, from the standard semantics of modal equation systems,
  and by Lemma~\ref{lem:stfsh-characterisation}, it follows that the
  least solution~$\mathcal{V}_0$ of the modal equation system is equal
  (on the names of the procedures called recursively by~$m$) to the
  valuation~$\mathcal{V}_{\rho_0}$ induced by the
  interpretation~$\rho_0$ defined by the procedure table~$T$.
  Finally, by Proposition~\ref{prop:call-unfolding-2} and
  Lemma~\ref{lem:stfsh-characterisation}, we have:
  \[
    \begin{array}{cl}
      & \fdenot{\stfh{\varnothing}{m()}} \\
      = & \logdenotV{X_m}{\mathcal{V}_0} \\
      = & \logdenotV{X_m}{\mathcal{V}_{\rho_0}} \\
      = & \mathcal{V}_{\rho_0} (X_m) \\
      = & \rho_0 (m) \\
      = & \trcdenoth{m()}_{\rho_0}  \\
      = & \trcdenot{m()}
    \end{array}
  \]
\end{proof}

\subsection{Soundness of the Calculus}
\label{sec:soundness-app}

For the proof we need to relate traces restricted by a condition $b$
to trace formulas. In the following proposition, the intuition for
$\fdenot{\lnot b\lor\phi}$ is that it ignores any trace, that is, it
is trivially true for any trace, where $b$ does not hold in the
beginning.

\begin{proposition}[State formulas in judgments]
  \label{prop:state-judge}
  Let $b$ be a Boolean expression, $\phi$ a trace formula. Then
  $
    (\trcdenot{S})\!\!\mid_b\; \subseteq \fdenot{\phi}
    \quad\text{iff}\quad
    \trcdenot{S} \subseteq \fdenot{\lnot b\lor\phi}.
  $  
\end{proposition}

\begin{proof}
  ``Only If'' direction:
  Assume $(\trcdenot{S})\!\!\mid_b\;\subseteq \fdenot{\phi}$ and there
  is a trace $s\cdot\sigma\in\trcdenot{S}$ that is not in
  $\fdenot{\lnot b\lor\phi}=\fdenot{\lnot b}\cup\fdenot{\phi}$, hence,
  $s\cdot\sigma\not\in\fdenot{\lnot b}$ and
  $s\cdot\sigma\not\in\fdenot{\phi}$. But
  $s\cdot\sigma\not\in\fdenot{\lnot b}$ implies
  $s\cdot\sigma\in(\trcdenot{S})\!\!\mid_b\;\subseteq\fdenot{\phi}$:
  contradiction.

  ``If'' direction:
  Assume $\trcdenot{S} \subseteq \fdenot{\lnot b\lor\phi}$ and there
  is a trace $s\cdot\sigma\in(\trcdenot{S})\!\!\mid_b$ that is not in
  $\fdenot{\phi}$. From the assumption and
  $(\trcdenot{S})\!\!\mid_b\subseteq\trcdenot{S}$ we obtain
  $s\cdot\sigma\in\fdenot{\lnot b\lor\phi}$.  However, since
  $\bexdenot{b} (s) = \ttt$ we must have in fact
  $s\cdot\sigma\in\fdenot{\phi}$: contradiction.
\end{proof}

\begin{proof}[Proof of Theorem~\ref{thm:calculus-soundness}]
  The proof system is sound, since every rule of the system is
  \emph{locally sound}, in the sense that its conclusion is valid
  whenever all its premises are valid.
  We shall prove local soundness of each rule.
  Without loss of generality, we ignore~$\Gamma$ in most cases.

  \noindent\textbf{Rule} \textsc{If}.

  Let~$\mathcal{I}$ be an arbitrary interpretation.
  Using Proposition~\ref{prop:state-judge}, we have: 
  \[
    \begin{array}{ll}
      & 
        \semsequenti{}{\judge{\Skip; S_1}{\lnot b \lor \phi}} \>\>\text{and}\>
        \semsequenti{}{\judge{\Skip; S_2}{b \lor \phi}} \\
      \Leftrightarrow\> & 
                          \trcdenoti{\Skip; S_1} \subseteq \fdenot{\lnot b \lor \phi}
                          \>\text{and}\>\\
      &
        \trcdenoti{\Skip; S_2} \subseteq \fdenot{b \lor \phi}  \\
      \Leftrightarrow & 
                        \sharp \trcdenoti{S_1} \subseteq \fdenot{\lnot b \lor \phi} 
                        \>\text{and}\>\>
                        \sharp \trcdenoti{S_2} \subseteq \fdenot{b \lor \phi}  \\
      \Leftrightarrow & 
                        (\sharp \trcdenoti{S_1}) \restr{b}{} \,\subseteq \fdenot{\phi}
                        \>\text{and}\>
                        (\sharp \trcdenoti{S_2}) \restr{\lnot b}{} \,\subseteq \fdenot{\phi} \\
      \Leftrightarrow & 
                        \left((\sharp \trcdenoti{S_1}) \restr{b}{} \cup\
                        (\sharp \trcdenoti{S_2}) \restr{\lnot b}{}\right)
                        \,\subseteq\,
                        \fdenot{\phi}  \\
      \Leftrightarrow\> & 
                          \trcdenoti{\If\ b\ \Then\ S_1\ \Else\ S_2} \subseteq \fdenot{\phi} \\
      \Leftrightarrow & 
                        \semsequenti{}{\judge{\If\ b\ \Then\ S_1\ \Else\ S_2}{\phi}} 
    \end{array}
  \]

  \noindent
  where we use that
  $\trcdenoti{\Skip; S} = \sharp \trcdenoti{S}$, and therefore:

  \[
    \begin{array}{ll}
     & \semsequenti{}{\judge{\Skip; S_1}{\lnot b \lor \phi}} \>\>\text{and}\>
      \semsequenti{}{\judge{\Skip; S_2}{b \lor \phi}}\\
  \Leftrightarrow\> & 
      \semsequenti{}{\judge{\If\ b\ \Then\ S_1\ \Else\ S_2}{\phi}}
    \end{array}
  \]

\noindent\textbf{Rule} \textsc{Call}.
For the proof of the rule we employ the principle of Fixed-Point
Induction. 
To simplify the presentation, we assume there is only one procedure
$m()$, declared as $m\, \{S_m\}$ in~$T$.
The general case follows from Beki\v{c}'s Principle. The notation
$\rho [m\mapsto\gamma]$ specifies the interpretation that is identical
to $\rho$, except for $\rho(m)=\gamma$.
%
%
\[
  \begin{array}{ll}
   & 
   \semsequent{Y_m : \phi_m}{\judge{S_m [\Skip; Y_m / m()]}{\phi_m}} \\
   \Leftrightarrow\> & 
   \forall \mathcal{I}.\
                       (\trcdenoti{Y_m} \subseteq \fdenot{\phi_m}\,\\
    & \qquad\Rightarrow 
       \trcdenoti{S_m [\Skip; Y_m / m()]} \subseteq \fdenot{\phi_m}) \\
   \Leftrightarrow\> & 
   \forall \gamma.\
      (\gamma \subseteq \fdenot{\phi_m}\, \Rightarrow 
       \trcdenoth{S_m}_{\rho [m \mapsto \sharp \gamma]} 
       \subseteq \fdenot{\phi_m}) \\
   \Leftrightarrow\> & 
   \forall \gamma.\
      (\sharp \gamma \subseteq \fdenot{\mathit{Id} \chop \phi_m}\, \Rightarrow 
       \sharp \trcdenoth{S_m}_{\rho [m \mapsto \sharp \gamma]} 
       \subseteq \fdenot{\mathit{Id} \chop \phi_m}) \\
   \Leftrightarrow\> & 
   \forall \gamma.\
      (\gamma \subseteq \fdenot{\mathit{Id} \chop \phi_m}\, \Rightarrow 
       \sharp \trcdenoth{S_m}_{\rho [m \mapsto \gamma]}\, 
       \subseteq \fdenot{\mathit{Id} \chop \phi_m}) \\
   \Rightarrow\> & 
   \rho_0 (m) \subseteq \fdenot{\mathit{Id} \chop \phi_m} \\
   %
   %
   \Leftrightarrow\> & 
   \trcdenoth{m()}_{\rho_0} \subseteq \fdenot{\mathit{Id}\chop\phi_m} \\
   \Leftrightarrow\> & 
   \trcdenot{m()} \subseteq \fdenot{\mathit{Id}\chop\phi_m} \\
   \Leftrightarrow\> & 
   \semsequent{}{\judge{m()}{\mathit{Id}\chop\phi_m}} 
  \end{array}
\]
\end{proof}

\subsection{Completeness of the Calculus}
\label{sec:completeness-app}

\begin{proof}[Proof of Theorem~\ref{thm:existence-rec}]
  The proof proceeds by induction on the structure of~$S$. However, In
  the case of a call the statement does not necessarily get smaller,
  because the body of~$m$ is expanded. So we need to argue that the
  induction is well-founded. Indeed, the lexicographic order on
  $\langle N-\left|\Gamma\right|,\left|S\right|\rangle$ (obviously,
  the first component is never negative) always decreases. Since
  $\Gamma$ is irrelevant for all cases except $S=m()$, we simplify the
  claim accordingly for these.

  \noindent\textbf{Case $S = \If\ b\ \Then\ S_1\ \Else\ S_2$.}
  Assume by the induction hypothesis that
  $\sequent{}{\judge{S_1}{\stfs{S_1}}}$ and
  $\sequent{}{\judge{S_2}{\stfs{S_2}}}$ can be proven.
  We also use that $\phi\models(p\lor(\lnot p\land\phi))$ is a valid
  consequence.
  We have: 
  \[
    \begin{array}{rl}
      & 
        \sequent{}{\judge{\If\ b\ \Then\ S_1\ \Else\ S_2}
        {\stfs{\If\ b\ \Then\ S_1\ \Else\ S_2}}} \\
      \Leftrightarrow & 
                        \sequent{}{\judge{\If\ b\ \Then\ S_1\ \Else\ S_2}
                        {(b \land \mathit{Id} \chop \stfs{S_1})
                        \lor
                        }} \\
      & \qquad\qquad\qquad\qquad\qquad(\lnot b \land \mathit{Id} \chop \stfs{S_2})\\
      \Leftarrow &
                   \mbox{\{By rule (\textsc{If}) combined with rule (\textsc{Or})\}} \\
      &
        \sequent{}{\judge{\Skip; S_1}
        {\lnot b \lor (b \land \mathit{Id} \chop \stfs{S_1})}} 
        ~\mbox{and}\\
      &
        \sequent{}{\judge{\Skip; S_2}
        {b \lor (\lnot b \land \mathit{Id} \chop \stfs{S_2})}} \\  
      \Leftarrow &
                   \mbox{\{By rule (\textsc{Cons})\}} \\
      &
        \sequent{}{\judge{\Skip; S_1}
        {\mathit{Id} \chop \stfs{S_1}}} 
        ~\mbox{and}
        \sequent{}{\judge{\Skip; S_2}
        {\mathit{Id} \chop \stfs{S_2}}} \\  
      \Leftrightarrow &
                        \mbox{\{By rule (\textsc{Seq}) combined with rule (\textsc{Skip})
                        and}\\
      & \mbox{the induction hypothesis\}} \\
      &
        \mathsf{true} 
    \end{array}
  \]
  
  \noindent\textbf{Case $S = m()$.}
  There are two subcases: either $m\in\{m_1,\ldots,m_n\}$ or not.
  In the first case,
  $S[\Skip;Y_{m_1}/m_1(),\ldots,\Skip;Y_{m_n}/m_n()]=\Skip;Y_m$

  In the second case,
  $S[\Skip;Y_{m_1}/m_1(),\ldots,\Skip;Y_{m_n}/m_n()]=m()$.
  In both cases, we have 
  $\stfs{m()} = \mathit{Id} \chop \phi_m$.
  
  \noindent\textbf{Subcase $\Skip;Y_m$} 

  We have to prove the judgment
  $\sequent{\Gamma}{\judge{\Skip;Y_m}{\mathit{Id} \chop \phi_m}}$.
  Using rules (\textsc{Seq}) and (\textsc{Skip}), this reduces to
  $\sequent{\Gamma}{\judge{Y_m}{\phi_m}}$.  Because of the assumption
  $m\in\{m_1,\ldots,m_n\}$ we have $\judge{Y_m}{\phi_m}\in\Gamma$, so
  the proof is finished.
  
  \noindent\textbf{Subcase $m()$}

  We have to prove the judgment
  $\sequent{\Gamma}{\judge{m()}{\mathit{Id}\chop\phi_m}}$.  Rule
  (\textsc{Call}) is applicable due to assumption
  $m\not\in\{m_1,\ldots,m_n\}$.
  %
  %
  Let
  \[
    S_m'\defeq S_m[\Skip;Y_m/m(),\Skip;Y_{m_1}/m_1(),\ldots,\Skip;Y_{m_n}/m_n()]
  \]
  The obtained premise yields the new claim to prove:
  \[
    \sequent{\Gamma\cup\{\judge{Y_m}{\phi_m}\}}{\judge{S_m'}{\mu
        X_m.\,\stfh{\{m\}}{S_m}}}
  \]

  We apply rule (\textsc{Unfold}) on the right and obtain:
  \[
    \sequent{\Gamma\cup\{\judge{Y_m}{\phi_m}\}}{\judge{S_m'}{\stfh{\{m\}}{S_m}[\phi_m/X_m]}}
  \]

  By induction and by the definition of $\stfs{S}$ we finish the proof
  up to subgoals of the form (for some $m'\neq m$):
  \[
    \sequent{\Gamma\cup\{\judge{Y_m}{\phi_m}\}}{\judge{m'()}{\stfh{\{m\}}{m'}[\phi_m/X_m]}}
  \]
  
  Unfortunately, the structure of the formula on the right does not
  conform to $\stfs{m'}$ as required.
  But, observing that
  $\fdenot{\stfs{S}} = \fdenot{\stfh{\overline{X}}{S}}$,
  as well as soundness of unfolding, we can use \textsc{Cons} to
  obtain:
  \[
    \sequent{\Gamma\cup\{\judge{Y_m}{\phi_m}\}}{\judge{m'()}{\stfs{m'}}}
  \]  
  This follows from the induction hypothesis, because of
  $N-\left|\Gamma\right|>N-\left|\Gamma\cup\{\judge{Y_m}{\phi_m}\}\right|$.

\end{proof}

\subsection{Canonical Programs}
\label{sec:canonical-programs-app}

\begin{proof}[Proof of Proposition~\ref{prop:well}]
  We need to show that for any call $m_X()$ in $\can{\phi}$ there is
  exactly one declaration of $m_X$ in $T_\phi$. This is a
  straightforward structural induction. 
\end{proof}

\begin{lemma}
\label{lem:can-characterisation-rec}
  Let $\phi$ be an open trace formula not containing fixed-point 
  binders~$(\mu X.)$, and let
  $\can{\phi} = \langle S_\phi, T_\phi \rangle$.
  Then, we have:
  $ \trcdenoth{S_\phi}_\rho \,\widetilde{=}\ 
     \logdenotV{\phi}{\mathcal{V}_\rho} $
  for all interpretations $\rho : M \rightarrow 2^{\State^+}$
  of the procedures~$M$ declared in~$T_\phi$, and (induced) valuations 
  $\mathcal{V}_\rho : \mathsf{RVar} \rightarrow 2^{\State^+}$
  defined by
  $\mathcal{V}_\rho (X_m) \defeq \rho (m)$.
\end{lemma}

\begin{proof}
  We proceed by structural induction on $\phi$.

  \noindent\textbf{Case $\phi = p \land \psi$}
  
  By the induction hypothesis,
  $\trcdenoth{S_\psi}_\rho \widetilde{=}\,
  \logdenotV{\psi}{\mathcal{V}_\rho}$.
  \begin{align*}
    & \trcdenoth{\If\ p\ \Then\ S_\psi\ \Else\ \Abort}_\rho
      \\
    & \qquad = (\sharp\trcdenoth{S_\psi}_\rho)\restr{p}{}\cup\;(\sharp\trcdenoth{\Abort}_\rho)\restr{\neg p}{}\\
    & \qquad = (\sharp\trcdenoth{S_\psi}_\rho)\restr{p}{}\quad\quad
         \text{(Proposition~\ref{prop:any})}\\
    & \qquad = \State^+\restr{p}{}\cap\;\sharp\trcdenoth{S_\psi}_\rho \widetilde{=}\ \ \State^+\restr{p}{}\cap\;\trcdenoth{S_\psi}_\rho\\
    & \qquad \;\widetilde{=}\ \State^+\restr{p}{}\cap\,\logdenotV{\psi}{\mathcal{V}_\rho}\quad\text{(Induction hypothesis)}\\
    & \qquad = \logdenotV{p\land\psi}{\mathcal{V}_\rho}
  \end{align*}

  \noindent\textbf{Case $\phi = \phi_1 \lor \phi_2$} 

  By the induction hypothesis, we have that
  $\trcdenoth{S_{\phi_1}}_\rho \,\widetilde{=}\, 
   \logdenotV{\phi_1}{\mathcal{V}_\rho}$
  and
  $\trcdenoth{S_{\phi_2}}_\rho \,\widetilde{=}\, 
   \logdenotV{\phi_2}{\mathcal{V}_\rho}$.
  Therefore,
  \begin{align*}
    \trcdenoth{\If\ \ast\ \Then\ S_{\phi_1}\ \Else\ S_{\phi_2}}_\rho
    & =\, \sharp\trcdenoth{S_{\phi_1}}_\rho\cup\;\sharp\trcdenoth{S_{\phi_2}}_\rho\\
     \;\widetilde{=}\ \ \trcdenoth{S_{\phi_1}}_\rho\cup\;\trcdenoth{S_{\phi_2}}_\rho 
    & \;\widetilde{=}\ \logdenotV{\phi_1}{\mathcal{V}_\rho} \cup\ 
            \logdenotV{\phi_2}{\mathcal{V}_\rho}\quad\text{(Ind.\ hyp.)}\\
    & =\, \logdenotV{\phi_1\lor\phi_2}{\mathcal{V}_\rho}
  \end{align*}

  \noindent\textbf{Case $\phi = X$}

  We have:
  \begin{align*}
    \trcdenoth{m_X ()}_\rho
     = \rho (m_X) 
     = \mathcal{V}_\rho (X) 
     = \logdenotV{X}{\mathcal{V}_\rho}
  \end{align*}
\end{proof}

\begin{proof}[Proof of Theorem~\ref{thm:can-characterisation-rec}]
  By structural induction on~$\phi$.  We only sketch the proof idea
  here.  An argument can be made similar to the one in the last case
  of the proof of Theorem~\ref{thm:stf-characterisation-rec}, by
  referring to the modal equation system corresponding to~$\phi$, and
  using Lemma~\ref{lem:can-characterisation-rec} to generalise the
  treatment to open formulas~$\phi$ and statements that contain calls
  to procedures not declared in~$T_\phi$.  Proposition~\ref{prop:well}
  ensures that the program on the left is well-defined.
  Compositionality of the proof is justified by Beki\v{c}'s
  Principle. 
\end{proof}

\end{document}